\documentclass[pdflatex,sn-basic,Numbered]{sn-jnl}
\usepackage{amsthm, amsfonts, amsmath,amssymb,stmaryrd}
\usepackage{hyperref,graphicx,caption,tikz,comment,enumitem}
\usepackage{subcaption}
\usetikzlibrary{arrows.meta, positioning,calc}

\theoremstyle{thmstyleone}%
\newtheorem{theorem}{Theorem}
\newtheorem{proposition}[theorem]{Proposition}%
\newtheorem{lemma}[theorem]{Lemma}%
\newtheorem*{claim*}{Claim}

\theoremstyle{thmstyletwo}%
\newtheorem{example}{Example}%
\newtheorem{remark}{Remark}%

\theoremstyle{thmstylethree}%
\newtheorem{definition}{Definition}%

\newenvironment{claimproof}{\begin{proof}[Proof of claim]}{\phantom\qedhere\hfill$\blacksquare$\end{proof}}

\title{Eventual and Strong Eventual Notions in Public Announcements}
\date{\today}
\author*[1]{\fnm{Eiji} \sur{Yamada}}\email{yamada.e.112d@m.isct.ac.jp}
\affil*[1]{\orgdiv{Department of Mathematical and Computing Science}, \orgname{Institute of Science Tokyo}, \state{Tokyo}, \country{Japan}}

\abstract{In dynamic epistemic logic, the four notions of success, self-refutation, true lies, and impossible lies have been discussed in the context of public announcements. In this paper, we introduce eventual and strong eventual versions of these notions, as well as their transfinite versions, which allow transfinite iteration of announcements. We also introduce the notions of always informativeness when true or false. For example, a formula is eventually self-refuting if, whenever initially true, it eventually becomes false at some finite stage under iterated announcements, and strong eventual self-refutation further requires the formula to remain false at all sufficiently late stages. There are two main results. The first result gives the relationship among strong eventual notions, eventual notions, and several other conditions including conditions on the limit of the truth values of the announced formula, the uniform bound condition, and the fixed-point views of the Moore sentence and the self-fulfilling sentence. The second result gives the relationship among finite and transfinite versions of the eventual and strong eventual notions and the fixed-point views.}
\begin{document}

\maketitle
\section{Introduction}

The apparent absurdity of asserting the sentence ``\(p\), but I do not believe \(p\)'' was first noted by G. E. Moore in his ``Reply to My Critics''~\citep{Moore1942}.
Wittgenstein called this phenomenon \emph{Moore's paradox} and made it a central problem about belief and assertion~\citep{Wittgenstein1953}. Hintikka gave it a seminal treatment in epistemic logic
\citep{Hintikka1962}.  In that setting, the sentence is formalized as the \emph{Moore sentence} $\varphi:= p\land\neg\Box p$ where $\Box $ is interpreted as knowledge or belief. 

From a static perspective, the Moore sentence $\varphi$ is \emph{unknowable} in the sense that $\varphi\land\Box\varphi$ is unsatisfiable\footnote{The terms knowable and believable as the satisfiability of $\varphi\land\Box\varphi$ and $\Box\varphi$, respectively, are due to \citet{Yamada2026Unknowability}, which are static definitions. The terms unknowable and unbelievable are their negation, and they are known to be equivalent in KD45 and S5. To be precise, this definition of unknowability does not take the justification aspect of knowledge into consideration, and it rather means that there is no situation in which both $\varphi$ is true and the agent believes $\varphi$. \citep{Balbiani2008} uses the same term to mean that there is a formula $\psi$ such that after announcing $\psi$, $K\varphi$ holds where $K$ means knowledge, which is dynamic.}: In fact, we have
\[
    \varphi\land\Box\varphi\leftrightarrow (p\land\lnot\Box p)\land\Box (p\land\lnot\Box p)\leftrightarrow (p\land\lnot\Box p)\land\Box p\land\Box\lnot\Box p\leftrightarrow \bot. 
\]
In KD45 (i.e., if agents have belief consistency, positive and negative introspections), we have the validity $\Diamond\Box\psi\leftrightarrow \Box\psi$ so that $\varphi$ is also \emph{unbelievable} in the sense that $\Box\varphi$ is unsatisfiable. 

We can also view the Moore sentence from a dynamic perspective. \emph{Public announcements} are one of the main topics in dynamic epistemic logic. A public announcement of \(\varphi\) is an announcement to a group of agents that \(\varphi\) holds. In this setting, the Moore sentence $\varphi=p\land\lnot\Box p$ is the most typical example of \emph{self-refuting} formulas, formulas that become false after being announced whenever they are initially true. In fact, whenever \(\varphi\) is true, after its announcement \(\Box p\) holds, which in turn makes \(\varphi\) itself false.

As a closely related yet somewhat separate topic, the notion of \emph{success} has been a central topic in the study of public announcements. A formula $\varphi$ is successful if, whenever $\varphi$ is true, it remains true after being announced. This notion matters since success guarantees that true information is shared with others without changing its truth, which is often the purpose of announcements.

\citet{Holliday2010-HOLMPI-2} proved that in single-agent S5, the notions of \emph{always informativeness},
\emph{non-Cartesianness} (the unsatisfiability of $\Box\varphi$), and \emph{eventual
self-refutation} are equivalent. \citet{Yamada2026Unknowability} then interpreted non-Cartesianness of $\varphi$ as unbelievability of $\varphi$ and showed that  in single-agent KD45 and S5, the static notions of unknowability and unbelievability are equivalent to two different views of Moorean phenomena: One is the validity of $\varphi\leftrightarrow (\varphi\land\lnot\Box\varphi)$ and the other is a contradiction among $\varphi$ itself, the belief part and the possibility part of the disjunctive normal form of $\varphi$. Consequently in single-agent S5, those static conditions are equivalent to the dynamic notions of always informativeness and eventual self-refutation. 

On the other hand, \citep{AgotnesVanDitmarschWang2018} introduced the remaining cases, \emph{true lies} and \emph{impossible lies}, where $\varphi$ is a true lie if, whenever $\varphi$ is false, it becomes true after an announcement, and $\varphi$ is an impossible lie if, whenever $\varphi$ is false, it remains false after an announcement (see Table~\ref{tab:four_notions}). They extended these four notions to \emph{$\sigma$-validity} for a finite or infinite sequence $\sigma$ of $0$s and $1$s. For example, $\varphi$ is $101$-valid if, whenever initially true, it becomes false after the first announcement and becomes true after the second announcement. \citep{Yamada2026} then gave a complete classification of such sequences in terms of $\sigma$-validity in multi-agent K45, KD45, and S5.
\begin{table}[htbp]
  \centering
  \begin{tabular}{|c|c|}
    \hline
    Successful & $\models\varphi\to [\uparrow\varphi]\varphi$\\\hline
    Self-refuting & $\models\varphi\to [\uparrow\varphi]\lnot\varphi$\\\hline
    True lie & $\models\lnot\varphi\to [\uparrow\varphi]\varphi$\\\hline
    Impossible lie & $\models\lnot\varphi\to [\uparrow\varphi]\lnot\varphi$\\\hline 
  \end{tabular}
  \caption{Four notions}
   \label{tab:four_notions}
\end{table}

The idea of repeatedly announcing formulas (i.e., \emph{iterated announcements}) is not new. In fact, iterated
modal relativization has been studied as an epistemic program
\citep{Miller2005-MILTUO-11}; repeated updates have been approached as
dynamical systems with their own recurrence and limiting behavior
\citep{Sadzik2006,KleinRendsvig2017}; and repeated public announcements yield
fixed-point procedures in epistemic analyses of games
\citep{vanBenthem2007}.

In this paper, we introduce \emph{eventual} and \emph{strong eventual} versions of the above four notions. We also introduce \emph{transfinite} versions for these eventual and strong eventual notions. For example, we say that a formula $\varphi$ is eventually self-refuting (denoted $E_{10}$) iff whenever $\varphi$ is initially true, $\varphi$ eventually becomes false at some point when repeatedly announced. Also, $\varphi$ is strongly eventually self-refuting $(SE_{10})$ iff whenever $\varphi$ is initially true, $\varphi$ eventually becomes false forever at some point when repeatedly announced. Finally, transfinitely eventually self-refuting ($E_{10}^{Ord}$) and transfinitely strongly eventually self-refuting ($SE_{10}^{Ord}$) refer to their transfinite counterparts, allowing transfinitely repeated announcements through ordinals. We also introduce the notions of always informativeness when true and always informativeness when false, the former of which is the same as always informativeness as introduced by \citep{Holliday2010-HOLMPI-2}.

There are two main results in this paper. Theorem~\ref{thm:finite-classification} gives, for each of the above four notions, the logical relationships among strong eventual notions, eventual notions, and several other conditions. Always informativeness when true/false are also included for self-refutation/true lies, respectively. According to the theorem, strong eventual notions are characterized by the limit of the truth values of the announced formula. Eventual notions are characterized by (i) limit superior or limit inferior of the truth values and (ii) the uniform bound condition. 

Condition (ii) claims that for any $E_{ij}$-formula $\varphi$ (where $i,j\in\{0,1\}$), one can take a uniform bound $N\geq 1$ across all pointed models such that whenever the truth of $\varphi$ is $i$, that of $\varphi$ eventually becomes $j$ by the $N$-th announcement. Furthermore, for $(i,j)=(1,0)$ (i.e., eventual self-refutation), all the agents' beliefs are destroyed and 
become inconsistent by the $N$-th announcement. 

  Eventual notions imply the fixed-point views of the sentences $\varphi\land\lnot C_G\varphi$ and $\varphi\lor C_G\varphi$ (denoted $S_{ij}^{<\omega}$). These conditions suggest that what underlies self-refutation and impossible lies is the generalized form $\varphi\land\lnot C_G\varphi$ of the Moore sentence $p\land\lnot\Box p$ while what underlies success and true lies is the generalized form $\varphi\lor C_G\varphi$ of the \emph{self-fulfilling sentence} $p\lor\Box p$. The theorem also implies that while every displayed one-way direction is strict when there are at least two agents, all the conditions collapse to the same notion in the single-agent case.

Theorem~\ref{thm:transfinite-classification} gives, for each $(i,j)\in \{0,1\}^2$, the logical relationship among $E_{ij}$, $SE_{ij}$, $E_{ij}^{Ord}$, $SE_{ij}^{Ord}$, $S_{ij}^{<\omega}$, and $S_{ij}^{Ord}$ from the three perspectives: strong eventual vs eventual, finite iterated announcements vs transfinite iterated announcements, and what condition $S_{ij}^{Ord}$ is equivalent to.

This paper is structured as follows.
Section~\ref{sec:basic-definitions-properties} introduces definitions of the usual epistemic logic, public announcement logic, believed public announcement logic, the four notions, and their basic properties. Section~\ref{sec:eventual-strong-eventual-notions} defines eventual and strong eventual notions. Section~\ref{sec:lemmas-for-main-results} proves various lemmas for our main results.
Section~\ref{sec:main-results} gives our main results and their interpretations. Section~\ref{sec:conclusion} gives a conclusion and future work.

A Lean formalization is available at \url{https://github.com/eiyamada/lean-paper-formalizations}.

\section{Basic definitions and properties}\label{sec:basic-definitions-properties}
We first list several basic definitions. Let $G$ be a nonempty finite set of agents and $\mathbf{Prop}$ be a countably infinite set of proposition letters.
\begin{definition}
    Define \textit{formulas} in the multi-agent epistemic logic $\mathcal{L}$ by
    \begin{equation*}
        \varphi:=p\mid\lnot\varphi\mid\varphi\land\psi\mid\Box_i\varphi\quad(i\in G,\,p\in \mathbf{Prop})
    \end{equation*}
\end{definition}
\begin{definition}[\citep{plaza2007}]
    Define \textit{formulas} in public announcement logic $\mathcal{L}_{PAL}$ by
    \begin{equation*}
        \varphi:=p\mid\lnot\varphi\mid \varphi\land\psi\mid \Box_i\varphi\mid [!\varphi]\psi\quad (i\in G,\,p\in\mathbf{Prop})
    \end{equation*}
\end{definition}
\begin{definition}[\citep{AgotnesVanDitmarschWang2018,GerbrandyGroeneveld1997}]
    Define \textit{formulas} in believed public announcement logic $\mathcal{L}_{\text{BPAL}}$ by
    \begin{equation*}
        \varphi:=p\mid\lnot\varphi\mid \varphi\land\psi\mid \Box_i\varphi\mid [\uparrow\varphi]\psi\quad (i\in G,\,p\in\mathbf{Prop})
    \end{equation*}
\end{definition}

We let $\mathcal{L}_C$ and $\mathcal{L}_{BPALC}$ denote the languages that add the common belief operator $C_G$ to $\mathcal{L}$ and $\mathcal{L}_{BPAL}$, respectively.
\begin{definition}
    A \textit{model} is a tuple $M=(W,\{R_i\}_{i\in G},V)$ where $W$ is a non-empty set of \textit{states}, $R_i$ is a binary relation on $W$ (\textit{accessibility relation}), and $V\colon \mathbf{Prop}\to 2^W$ is a \textit{valuation function}.
\end{definition}
For a formula $\varphi$ and a pointed model $M,w$, let $\llbracket\varphi\rrbracket_M:=\{w\in W\colon M,w\models\varphi\}$ be the set of states in which $\varphi$ holds.

\begin{definition}\label{def:relativized model}
    Let $M=(W, \{R_i\}_{i\in G}, V)$ be a model and $\varphi$ be a formula.
    \begin{itemize}
        \item When $\llbracket\varphi\rrbracket_M\neq\varnothing$, the \emph{relativization} of $M$ to $\varphi$ under public announcement is the model $M_{|\varphi}=(W_{|\varphi}, \{R_{i|\varphi}\}_{i\in G}, V_{|\varphi})$ where $W_{|\varphi}=\llbracket\varphi\rrbracket_M$, $R_{i|\varphi}=R_i\cap (W_{|\varphi}\times W_{|\varphi})$, and $V_{|\varphi}(p)=V(p)\cap W_{|\varphi}$ for all $p\in\mathbf{Prop}$.
        \item The \emph{relativization} of $M$ to $\varphi$ under believed public announcement is the model $M|\varphi=(W, \{R_i|\varphi\}_{i\in G}, V)$ with $R_i|\varphi=R_i\cap (W\times \llbracket \varphi\rrbracket_M)$.
    \end{itemize}
\end{definition}
\begin{definition}
    Let $M=(W,\{R_i\}_{i\in G},V)$ be a model. Define truth as follows.
    \begin{enumerate}
        \item $M,w\models p\iff w\in V(p)$.
        \item $M,w\models\lnot\varphi\iff M,w\not\models\varphi$.
        \item $M,w\models\varphi\land\psi\iff M,w\models\varphi\text{ and }M,w\models\psi$.
        \item $M,w\models\Box_i\varphi\iff M,v\models\varphi$ for all $v$ with $wR_iv$.
        \item $M,w\models[!\varphi]\psi\iff M,w\models\varphi\Rightarrow M_{|\varphi}, w\models\psi$.
        \item $M,w\models[\uparrow\varphi]\psi\iff M|\varphi, w\models\psi$.
        \item $M,w\models C_G\varphi\iff M,v\models\varphi$ for all $v$ with $w(\bigcup_{i\in G}R_i)^+v$ where $+$ denotes the transitive closure.
    \end{enumerate}
\end{definition}

Although public announcement logic (PAL) has been widely used in the literature, in this paper we use believed public announcement logic (BPAL) since in PAL, we can no longer consider the truth of formulas at states that have already been eliminated. That is, when we consider the truth of a formula $\varphi$ at the pointed model $M,w$, we can no longer consider the truth at $w$ in the updated model $M_{|\varphi}$ whenever $M,w\models\lnot\varphi$, simply because $w\notin M_{|\varphi}$.

The expressive powers of $\mathcal{L}$, $\mathcal{L}_{PAL}$, and $\mathcal{L}_{BPAL}$ are known to be the same. That is, for any formula in one language, there is a logically equivalent formula in the other language. In fact, for example, the following lemma gives reduction axioms from $\mathcal{L}_{BPAL}$ to $\mathcal{L}$. The proof uses an induction on formulas.

\begin{lemma}[Reduction axioms for BPAL]\label{lem:bpal-reduction-axioms}
    For all formulas $\varphi,\psi,\chi$, all proposition letters $p$, and all agents $i\in G$, the following formulas are valid:
    \begin{align*}
        [\uparrow\varphi]p
        &\leftrightarrow p,\\
        [\uparrow\varphi]\neg\psi
        &\leftrightarrow \neg[\uparrow\varphi]\psi,\\
        [\uparrow\varphi](\psi\land\chi)
        &\leftrightarrow
        ([\uparrow\varphi]\psi\land[\uparrow\varphi]\chi),\\
        [\uparrow\varphi]\Box_i\psi
        &\leftrightarrow
        \Box_i(\varphi\to[\uparrow\varphi]\psi),\\
        [\uparrow\varphi]\Diamond_i\psi
        &\leftrightarrow
        \Diamond_i(\varphi\land[\uparrow\varphi]\psi)\\
        [\uparrow\varphi][\uparrow\psi]\chi
        &\leftrightarrow
        [\uparrow(\varphi\land[\uparrow\varphi]\psi)]\chi.
    \end{align*}
\end{lemma}
Let $R^*$ be the reflexive transitive closure of a binary relation $R$.
\begin{definition}
    A model $M'=(W',\{R_i'\}_{i\in G},V')$ is a \textit{submodel} of $M=(W,\{R_i\}_{i\in G},V)$ (written $M'\subseteq M$) iff $W'\subseteq W$, $R_i'=R_i\cap (W'\times W')$ for all $i\in G$, and $V'(p)=V(p)\cap W'$ for all $p\in\mathbf{Prop}$. The \textit{generated submodel} of $M$  at $w$ is the submodel $M_w$ with domain $\{v\in W\colon w(\bigcup_{i\in G}R_i)^*v\}$.
\end{definition}
The following lemma states that in K45, the truth values of ``modal atoms'' do not change before and after moving between two states.

\begin{lemma}\label{lem:modal agreement lemma}
    Let $M=(W,\{R_i\}_{i\in G},V)$ be a K45 model. If $wR_iv$, then for any formula $\varphi$ of the form $\Box_i\psi$ or $\Diamond_i\psi$, we have
    \begin{equation*}
        M,w\models\varphi\Longleftrightarrow M,v\models\varphi
    \end{equation*}
\end{lemma}
\begin{proof}
    For $\varphi=\Box_i\psi$, $(\Rightarrow)$ follows from transitivity and $(\Leftarrow)$ follows from Euclideanness. For $\varphi=\Diamond_i\psi$, the order of the properties is reversed.
\end{proof}
From the perspective of model stabilization under believed public announcements, transfinite iteration is a natural extension of finite iteration. Beyond its mathematical role, it may also admit meaningful epistemic or communicative interpretations, a possibility we leave for future investigation (cf.~\citep[Sections~1.2, 2.4, 5.1--5.4, and~6.2]{vanBenthem2007},
where higher-order reasoning about players' rationality motivates
iterative strategy elimination, and transfinite approximation
connects the resulting epistemic processes to fixed points
characterizing game-theoretic solution concepts such as
rationalizability).
\begin{definition}
    Let $M=(W,\{R_i\}_{i\in G},V)$ be a model. For each ordinal $\alpha$, recursively define the model $M|^\alpha\varphi=M^\alpha=(W,\{R_i|^\alpha\varphi\}_{i\in G},V)$ by:
    \begin{enumerate}
        \item $R_i|^0\varphi=R_i$
        \item $R_i|^{\alpha+1}\varphi=(R_i|^\alpha\varphi)|\varphi$
        \item For limit ordinal $\lambda$, $R_i|^\lambda\varphi=\bigcap_{\alpha<\lambda}R_i|^\alpha\varphi$.
    \end{enumerate}
\end{definition}
The following proposition states that after sufficiently many (possibly transfinite) iterations, the model stabilizes and the announced formula becomes common belief. 
\begin{proposition}\label{prop:ordinal-stabilization}
    Let $\varphi$ be a formula. For all models $M$, there is an ordinal $\alpha$ such that $M^\alpha=M^{\alpha+1}$. Furthermore, this $\alpha$ satisfies $M^\alpha,w\models C_G\varphi$ for all $w\in M$.
\end{proposition}
\begin{proof}
    Let $E=\{(i,u,v)\colon i\in G,\,(u,v)\in R_i\}$ be the set of all labelled arrows of $M$. Put $\kappa=|E|$, and fix an enumeration without repetition $E=\{e_\xi:\xi<\kappa\}$.
For each ordinal $\beta$, let $E^\beta=\{(i,u,v):i\in G,\ (u,v)\in R_i^\beta\}$ be the labelled arrows remaining in $M^\beta$.
Suppose, for a contradiction, that $M^\beta\ne M^{\beta+1}$ for every
$\beta<\kappa^+$. Then $E^\beta\setminus E^{\beta+1}\ne\varnothing$, so define
\[f:\kappa^+\longrightarrow\kappa,\qquad
f(\beta)=\min\{\xi<\kappa:e_\xi\in E^\beta\setminus E^{\beta+1}\}.
\]
Thus, $f(\beta)$ is the index of the first arrow deleted at stage $\beta$.
To check injectivity, let $\beta<\gamma<\kappa^+$. Since the arrow sets decrease,
\[
e_{f(\gamma)}\in E^\gamma\subseteq E^{\beta+1},
\qquad
e_{f(\beta)}\notin E^{\beta+1}.
\]
Hence $f(\beta)\ne f(\gamma)$. This gives the impossible injection from
$\kappa^+$ into $\kappa$. Hence $M^\alpha=M^{\alpha+1}$ for some $\alpha<|E|^+$. $M^\alpha,w\models C_G\varphi$ is immediate from the transitive closure style definition of $C_G$.
\end{proof}

\begin{definition}[\citep{AgotnesVanDitmarschWang2018}]Let $\varphi$ be a formula.
    \begin{itemize}
        \item $\varphi$ is \textit{successful} iff $\varphi\to[\uparrow\varphi]\varphi$ is valid.
        \item $\varphi$ is \textit{self-refuting} iff $\varphi\to[\uparrow\varphi]\lnot\varphi$ is valid.
        \item $\varphi$ is a \textit{true lie} iff $\lnot\varphi\to[\uparrow\varphi]\varphi$ is valid.
        \item $\varphi$ is an \textit{impossible lie} iff $\lnot\varphi\to[\uparrow\varphi]\lnot\varphi$ is valid.
    \end{itemize}
\end{definition}
\begin{example}\label{ex:p-or-box-p}
    We show that $\varphi:=p\lor\Box p$ (``$p$ is true or the agent believes $p$'') is a true lie in K45. Suppose that $M,w\models\lnot(p\lor\Box p)$. Then, we have $M,w\models \lnot p\land\lnot\Box p$. Take any $v\in R|\varphi(w)$. Then, $M,v\models p\lor\Box p$ but by Lemma~\ref{lem:modal agreement lemma}, we also have $M,v\models\lnot\Box p$. Thus, we must have $M,v\models p$ so $M|\varphi,w\models\Box p$ hence $M|\varphi,w\models\varphi$. In this paper, we call this formula the \emph{self-fulfilling sentence}.
\end{example}

\section{Eventual and strong eventual notions}\label{sec:eventual-strong-eventual-notions}
In this section, we introduce, for each $i,j\in\{0,1\}$, eventual $(E_{ij})$ and strong eventual ($SE_{ij}$) variants and their transfinite versions ($E_{ij}^{Ord}$ and $SE_{ij}^{Ord}$). We also introduce always informativeness when true/false ($AI_i$) and the $S_{ij}$ conditions.
Unless otherwise stated, all models considered below are multi-agent K45 models, and validity and satisfiability are understood relative to this class.
\begin{definition}
For $b\in\{0,1\}$, write $\varphi^1:=\varphi$ and $\varphi^0:=\neg\varphi$. For $i,j\in\{0,1\}$:
When $\varphi$ is fixed, we abbreviate $M|^\alpha\varphi$ by $M^\alpha$.
\begin{align*}
 E_{ij}(\varphi)
 &:\Longleftrightarrow
 \forall M,w\,
 \bigl(M,w\models\varphi^i
 \Rightarrow
 \exists n\ge1\;M^n,w\models\varphi^j\bigr),
 \\
 SE_{ij}(\varphi)
 &:\Longleftrightarrow
 \forall M,w\,
 \bigl(M,w\models\varphi^i
 \Rightarrow
 \exists N\ge1\;\forall n\ge N\;M^n,w\models\varphi^j\bigr),
 \\
 E^{Ord}_{ij}(\varphi)
 &:\Longleftrightarrow
 \forall M,w\,
 \bigl(M,w\models\varphi^i
 \Rightarrow
 \exists\alpha>0\;M^\alpha,w\models\varphi^j\bigr),
 \\
 SE^{Ord}_{ij}(\varphi)
 &:\Longleftrightarrow
 \forall M,w\,
 \bigl(M,w\models\varphi^i
 \Rightarrow
 \exists\alpha>0\;\forall\beta\ge\alpha\;M^\beta,w\models\varphi^j\bigr),\\
 AI_i(\varphi) &:\Longleftrightarrow
 \forall M,w\,\bigl(M,w\models\varphi^i\Rightarrow
 (M|\varphi)_w\ne M_w\bigr).
\end{align*}
For $k<\omega$, write 
\begin{align*}
 S_{ij}(k,\varphi)
 &:\Longleftrightarrow
 \models
 \varphi^i\to[\uparrow\varphi]^k(C_G\varphi\to\varphi^j),\\
 S^{<\omega}_{ij}(\varphi)
 &:\Longleftrightarrow
 \forall k<\omega\;S_{ij}(k,\varphi),\\
 S^{Ord}_{ij}(\varphi)&:\Longleftrightarrow
 \forall M,w,\alpha\,
 \bigl(M,w\models\varphi^i\land M^\alpha,w\models C_G\varphi
 \Rightarrow M^\alpha,w\models\varphi^j\bigr)
\end{align*}
where $\alpha$ ranges over the ordinals. When the formula argument is omitted, the symbols above denote the corresponding classes of formulas.

We use the terms \emph{eventually successful}, \emph{strongly eventually successful}, \emph{transfinitely eventually successful}, \emph{transfinitely strongly eventually successful}, \emph{eventually self-refuting}, \emph{strongly eventually self-refuting}, \emph{transfinitely eventually self-refuting}, \emph{transfinitely strongly eventually self-refuting}, \emph{eventual true lie}, \emph{strong eventual true lie}, \emph{transfinite eventual true lie}, \emph{transfinite strong eventual true lie}, \emph{eventual impossible lie}, \emph{strong eventual impossible lie}, \emph{transfinite eventual impossible lie}, \emph{transfinite strong eventual impossible lie}, \emph{always informative when true}, and \emph{always informative when false} in the obvious manner.

In K45, a formula $\varphi$ cannot be both always informative when true and always informative when false. In fact, according to Theorem~\ref{thm:finite-classification}, 
\begin{equation*}
AI_1(\varphi)\Longleftrightarrow\models C_G\varphi\to\lnot\varphi \qquad\text{and} \qquad AI_0(\varphi)\Longleftrightarrow\models C_G\varphi\to\varphi,
\end{equation*}
 so they together imply $\models\lnot C_G\varphi$. However, the single-state model with empty accessibility relations is indeed K45 and the state vacuously satisfies $C_G\varphi$. Note that the situation differs in KD45 and S5 since empty accessibility relations are not allowed: in such frames, we could possibly say something like ``$\varphi$ is commonly unbelievable/unknowable.''

Also, although it is possible to define ``$\varphi$ is transfinitely always informative when true/false'' as
\begin{equation*}
    \forall M,w,\alpha\quad(M^\alpha,w\models\varphi^i\Rightarrow (M^{\alpha+1})_w\neq (M^\alpha)_w),
\end{equation*}
this immediately reduces to the same notion as ``$\varphi$ is always informative when true/false'' since K45 is closed under updates.

Finally, for a formula $\varphi$, a pointed model $M,w$, and $n<\omega$, put
\[
\sigma_n^{M,w}(\varphi):=
\begin{cases}
1 & \text{if } M|^n\varphi,w\models\varphi,\\
0 & \text{if } M|^n\varphi,w\models\lnot\varphi.
\end{cases}
\]
For readability, we write $\sigma_n^{M,w}$ for
$\sigma_n^{M,w}(\varphi)$.
\end{definition}

\section{Lemmas for the main results}\label{sec:lemmas-for-main-results}
\begin{lemma}\label{lem:finite-eventual-facts}
    Let $\varphi$ be a basic multi-agent epistemic formula, $M$ be a model, and $w\in M$.
    
    \begin{enumerate}[label=(\arabic*)]
        \item $(M|\varphi)_w=M_w$ iff $M,w\models C_G\varphi$.
        \item For all $i,j\in\{0,1\}$, $E_{ij}(\varphi)$ implies $S_{ij}^{<\omega}(\varphi)$. Moreover, $AI_i(\varphi)$, $S_{i,1-i}^{<\omega}(\varphi)$, and $\models C_G\varphi\to\varphi^{1-i}$ are equivalent.
        \item For all $i,j\in\{0,1\}$,
$E_{ij}(\varphi)$ is equivalent to the existence of
$N\geq1$ such that
\[
    \models
    \varphi^i\to
    \bigvee_{n=1}^{N}
    [\uparrow\varphi]^n\varphi^j.
\]
        \item $E_{10}(\varphi)$, $SE_{10}(\varphi)$, and the following condition are equivalent: there is an $n\geq 1$ such that
              \begin{equation*}
                  \models[\uparrow\varphi]^n\left(\neg\varphi\land\bigwedge_{i\in G}\Box_i\bot\right).
              \end{equation*}
    \end{enumerate}
\end{lemma}
\begin{proof}
    (1) Suppose \(M,w\models C_G\varphi\). Every arrow in \(M_w\) has a target reachable from \(w\) by a nonempty \(G\)-path, and that target therefore satisfies \(\varphi\). Hence no arrow in \(M_w\) is deleted, so \((M|\varphi)_w=M_w\). Conversely, suppose \((M|\varphi)_w=M_w\). Every nonempty \(G\)-path from \(w\) in \(M\) then remains after the update. The target of its final arrow must therefore satisfy \(\varphi\) in \(M\). Thus \(M,w\models C_G\varphi\).

    (2) For $E_{ij}(\varphi)\Rightarrow S_{ij}^{<\omega}(\varphi)$, suppose $E_{ij}(\varphi)$ and suppose toward a contradiction that $S_{ij}(k,\varphi)$ fails. Then, there are $M,w$ such that $M,w\models\varphi^i$ and $M^k,w\models C_G\varphi\land\varphi^{1-j}$ for some $k$. By (1), the value $1-j$ is permanent from that stage on. If $i\neq j$, then $1-j=i$, so applying $E_{ij}$ to $M^k$ yields a contradiction. If $i=j$, let $m<k$ be the last stage before $k$ such that $M|^m,w\models \varphi^{i}$. Then, applying $E_{ii}$ yields a contradiction since $1-j\neq i$. Hence $E_{ij}$ implies every $S_{ij}(k)$. 
    
     Next, we show equivalence among $AI_i(\varphi)$, $S_{i,1-i}^{<\omega}(\varphi)$, and $\models C_G\varphi\to\varphi^{1-i}$. Recall that $AI_i(\varphi)\Longleftrightarrow
 \forall M,w\,\bigl(M,w\models\varphi^i\Rightarrow
 (M|\varphi)_w\ne M_w\bigr)$. Thus, by (1), $AI_i(\varphi)$ is equivalent to $\models\varphi^i\to\neg C_G\varphi$, which is propositionally equivalent to $\models C_G\varphi\to\varphi^{1-i}$. For $S_{i,1-i}^{<\omega}(\varphi)$ to $\models C_G\varphi\to\varphi^{1-i}$, note that $S_{i,1-i}^{<\omega}(\varphi)\Longleftrightarrow
 \forall k<\omega,\, \models
 \varphi^i\to[\uparrow\varphi]^k(C_G\varphi\to\varphi^{1-i})$. So, taking $k=0$ yields $\models\varphi^i\to ( C_G\varphi\to\varphi^{1-i})$, which is propositionally equivalent to $\models C_G\varphi\to\varphi^{1-i}$. Conversely, $\models C_G\varphi\to\varphi^{1-i}$ clearly gives $S_{i,1-i}^{<\omega}(\varphi)$ since K45 is closed under updates.

(3) 
The right-to-left implication is immediate from the definition of
$E_{ij}(\varphi)$.

For the converse, suppose $E_{ij}(\varphi)$ holds and that no such
uniform bound exists. Then, for every $N\geq1$, the formula
\[
    \varphi^i\land
    \bigwedge_{n=1}^{N}
    [\uparrow\varphi]^n\varphi^{1-j}
\]
is K45-satisfiable. Since BPAL is reducible to the basic epistemic
language, compactness of K45 yields a pointed model $M,w$ such that
$M,w\models\varphi^i$ and
$M^n,w\models\varphi^{1-j}$ for every $n\geq1$.
This contradicts $E_{ij}(\varphi)$.

    (4) We first show that $E_{10}(\varphi)$ implies $\models[\uparrow\varphi]^n\left(\neg\varphi\land\bigwedge_{i\in G}\Box_i\bot\right)$. Suppose $E_{10}$. By (3), there is an $N\geq 1$ such that
\begin{equation*}
\models
\varphi\to
\bigvee_{m=1}^{N}[\uparrow\varphi]^m\neg\varphi.
\end{equation*}

We now show $\models
[\uparrow\varphi]^{N+1}
\left(
\neg\varphi\land\bigwedge_{i\in G}\Box_i\bot
\right)$. Take any model $M=(W,\{R_i\}_{i\in G},V)$ and $w\in M$. We first show
$M,w\models[\uparrow\varphi]^{N+1}\bigwedge_{i\in G}\Box_i\bot$. That is, every agent's relation is empty after the $N+1$-th announcement. Suppose, towards a contradiction, that $x(R_i|^{N+1}\varphi)y$ for some $i\in G$ and $x,y\in M$. Then, since an arrow survives from stage $m$ to stage $m+1$ only if its target satisfies $\varphi$ at stage $m$, we get
$M|^m\varphi,y\models\varphi$
for every $0\leq m\leq N$. However, this contradicts the earlier claim, so that we have $\models[\uparrow\varphi]^{N+1}\bigwedge_{i\in G}\Box_i\bot$.

We finally show $M,w\models[\uparrow\varphi]^{N+1}\lnot\varphi$. Suppose toward a contradiction that
$M|^{N+1}\varphi,w\models\varphi$.
Since the model is edgeless, further announcements of $\varphi$ do not change it. Hence $M|^{N+1+m}\varphi,w\models\varphi$ for every $m<\omega$. However, applying $E_{10}(\varphi)$ to the pointed K45 model $M|^{N+1}\varphi,w$ gives a contradiction. Therefore, we have
\begin{equation*}
\models
[\uparrow\varphi]^{N+1}
\left(
\neg\varphi\land\bigwedge_{i\in G}\Box_i\bot
\right).
\end{equation*}

Conversely, suppose that 
$\models
[\uparrow\varphi]^n
\left(
\neg\varphi\land\bigwedge_{i\in G}\Box_i\bot
\right)$ for some $n\geq 1$.
Then after $n$ announcements the model is edgeless and $\varphi$ is false everywhere. Since an edgeless model is unchanged by every further believed public announcement, $\varphi$ remains false at all later stages. Thus $SE_{10}(\varphi)$ holds. Finally, $SE_{10}(\varphi)$ implies $E_{10}(\varphi)$ immediately from the definitions.
\end{proof}
\begin{lemma}\label{lem:ordinal-eventual-facts}
    Let $\varphi$ be a basic multi-agent epistemic formula. For every $i,j\in\{0,1\}$ we have $SE_{ij}\Rightarrow E_{ij}\Rightarrow E_{ij}^{Ord}\Rightarrow S_{ij}^{<\omega}$, $SE_{ij}^{Ord}\Rightarrow E_{ij}^{Ord}$, and $S_{ij}^{Ord}\Leftrightarrow SE_{ij}^{Ord}$. For $i\neq j$, $E_{ij}^{Ord}$, $SE_{ij}^{Ord}$, $S_{ij}^{Ord}$, and $S_{ij}^{<\omega}$ are equivalent.
\end{lemma}
\begin{proof}
Most of the displayed implications follow immediately from the definitions:
$SE_{ij}\Rightarrow E_{ij}\Rightarrow E_{ij}^{Ord}$ and
$SE_{ij}^{Ord}\Rightarrow E_{ij}^{Ord}$. Thus it remains to prove three facts:
$E_{ij}^{Ord}\Rightarrow S_{ij}^{<\omega}$,
$S_{ij}^{Ord}\Leftrightarrow SE_{ij}^{Ord}$, and, when $i\neq j$,
$S_{ij}^{<\omega}\Rightarrow SE_{ij}^{Ord}$.

\medskip\noindent\textbf{Proof of $E_{ij}^{Ord}\Rightarrow S_{ij}^{<\omega}$.}
We first prove $E_{ij}^{Ord}\Rightarrow S_{ij}^{<\omega}$. Fix $k<\omega$ and suppose, towards a contradiction, that $S_{ij}(k,\varphi)$ fails. Then there are a K45 model $M$ and a state $w$ such that
    $M,w\models\varphi^i$ and
    $M^k,w\models C_G\varphi\land\varphi^{1-j}$.
By Lemma~\ref{lem:finite-eventual-facts}(1), the generated model at $w$ is fixed from stage $k$ onward. Hence
\begin{equation*}
    M^\alpha,w\models\varphi^{1-j}
    \qquad\text{for every }\alpha\geq k.
\end{equation*}

Suppose first that $i\neq j$. Since $i,j\in\{0,1\}$, we have $i=1-j$, so $M^k,w\models\varphi^i$. Applying $E_{ij}^{Ord}(\varphi)$ to the pointed tail model $M^k,w$ requires $\varphi^j$ to hold at some later ordinal stage. This is impossible because $\varphi^{1-j}$ is permanent from stage $k$ onward.

Now suppose that $i=j$. Since $\varphi^i$ holds at stage $0$ and $\varphi^{1-i}$ holds at stage $k$, we have $k>0$. Let $m<k$ be the last finite stage before $k$ at which
$M^m,w\models\varphi^i$. Then
\begin{equation*}
    M^n,w\models\varphi^{1-i}
    \qquad\text{for every }m<n\leq k.
\end{equation*}
Moreover, $\varphi^{1-i}$ is permanent from stage $k$ onward by the preceding argument. Thus the tail beginning at $M^m,w$ never returns to the value $i$ at any positive ordinal stage, contradicting $E_{ii}^{Ord}(\varphi)$. Therefore $S_{ij}(k,\varphi)$ holds for every $k<\omega$, and hence $E_{ij}^{Ord}(\varphi)\Rightarrow S_{ij}^{<\omega}(\varphi)$.

\medskip\noindent\textbf{Proof of $S_{ij}^{Ord}\Leftrightarrow SE_{ij}^{Ord}$.}
We next prove $S_{ij}^{Ord}\Leftrightarrow SE_{ij}^{Ord}$. Suppose first that $S_{ij}^{Ord}(\varphi)$ holds, and let $M,w\models\varphi^i$. By Proposition~\ref{prop:ordinal-stabilization}, there is an ordinal $\lambda$ such that $M^\lambda=M^{\lambda+1}$ and
$M^\lambda,w\models C_G\varphi$. Hence, by $S_{ij}^{Ord}(\varphi)$, we have $M^\lambda,w\models\varphi^j$.
Since $M^\lambda$ is a fixed point, this truth value is permanent at all later stages. Thus $SE_{ij}^{Ord}(\varphi)$ holds. If $\lambda=0$, we may use stage $1$ as the positive witnessing ordinal.

Conversely, suppose that $SE_{ij}^{Ord}(\varphi)$ holds. Let $M,w\models\varphi^i$, and suppose that for some ordinal $\alpha$,
$M^\alpha,w\models C_G\varphi$. By Lemma~\ref{lem:finite-eventual-facts}(1), the generated model at $w$ is fixed from stage $\alpha$ onward. Hence the truth value of $\varphi$ at $w$ is also permanent from stage $\alpha$ onward. If
$M^\alpha,w\models\varphi^{1-j}$, then
    $M^\beta,w\models\varphi^{1-j}$
    for every $\beta\geq\alpha$,
which contradicts $SE_{ij}^{Ord}(\varphi)$. Therefore
$M^\alpha,w\models\varphi^j$, and so $S_{ij}^{Ord}(\varphi)$ holds.

Thus, we have proved
    $S_{ij}^{Ord}(\varphi)\Longleftrightarrow SE_{ij}^{Ord}(\varphi)$.

\medskip\noindent\textbf{Proof of the equivalence when $i\neq j$.}
Finally, suppose that $i\neq j$. Then $j=1-i$, and Lemma~\ref{lem:finite-eventual-facts}(2) gives
\begin{equation*}
    S_{ij}^{<\omega}(\varphi)
    \Longleftrightarrow
    \models C_G\varphi\to\varphi^j.
\end{equation*}
Assume $S_{ij}^{<\omega}(\varphi)$ and let $M,w\models\varphi^i$. Choose a fixed stage $\lambda$ by Proposition~\ref{prop:ordinal-stabilization}. Since
$M^\lambda,w\models C_G\varphi$, the displayed validity gives
$M^\lambda,w\models\varphi^j$. As the model is fixed from stage $\lambda$ onward, $\varphi^j$ remains true forever. Hence $SE_{ij}^{Ord}(\varphi)$ holds.

Combining this implication with
$SE_{ij}^{Ord}\Rightarrow E_{ij}^{Ord}\Rightarrow S_{ij}^{<\omega}$ and with
$S_{ij}^{Ord}\Leftrightarrow SE_{ij}^{Ord}$, we obtain, for $i\neq j$, $E_{ij}^{Ord}
    \Longleftrightarrow
    SE_{ij}^{Ord}
    \Longleftrightarrow
    S_{ij}^{Ord}
    \Longleftrightarrow
    S_{ij}^{<\omega}$.
\end{proof}
\begin{lemma}\label{lem:single-agent_collapse}
    In single-agent K45, $M_w|^n\varphi=M_w|^{n+1}\varphi$ for some $n\geq 1$. Consequently, $SE_{ij}$, $E_{ij}$, $SE_{ij}^{Ord}$, $E_{ij}^{Ord}$, $S_{ij}^{Ord}$, and $S_{ij}^{<\omega}$ are equivalent for every $i,j\in\{0,1\}$.
\end{lemma}
\begin{proof}
Let $M=(W,R,V),w$ be a pointed K45 model, and let $P$ be the finite set of proposition letters occurring in $\varphi$. Put  $S_n:=(R|^n\varphi)(w)$.

We first show that the generated model at $w$ reaches a fixed point after finitely many announcements. In a transitive and Euclidean frame, $xRy$ implies $R(x)=R(y)$ and $yRy$. Hence, for every $y\in S_0$, we have $R(y)=S_0$. Since believed public announcement only restricts the targets of arrows, induction on $n$ gives $R^{n}(y)=S_n$ for every $y\in S_0$.

Now let $x,y\in S_0$ have the same $P$-valuation. We claim that, for every $n<\omega$, $M^n,x\models\varphi$ iff $M^n,y\models\varphi$. More generally, the same holds for every subformula of $\varphi$. This follows by induction on the construction of formulas. The propositional and Boolean cases are immediate. For the modal case, $x$ and $y$ have the same successor set $S_n$, so, for example, $M^n,x\models\Box\psi$ iff every state in $S_n$ satisfies $\psi$, iff $M^n,y\models\Box\psi$.

Therefore, whenever $S_{n+1}\subsetneq S_n$, at least one entire $P$-valuation type disappears from the successor set. Indeed, if $y\in S_n\setminus S_{n+1}$, then $M^n,y\models\neg\varphi$, and every state in $S_n$ with the same $P$-valuation as $y$ also satisfies $\neg\varphi$ and is deleted as a target at the next update. Since there are at most $2^{|P|}$ $P$-valuation types, the successor set can strictly decrease only finitely many times. Hence there is some $k<\omega$ such that $S_k=S_{k+1}$. It follows that the generated model at $w$ is fixed from stage $k$ onward. By Lemma~\ref{lem:finite-eventual-facts}(1), $M^k,w\models C_G\varphi$.

We now prove the collapse of the six notions. Suppose first that $S_{ij}^{<\omega}(\varphi)$ holds and that $M,w\models\varphi^i$. Choose a finite fixed stage $k$ as above. Since $S_{ij}(k,\varphi)$ holds, we have $M^k,w\models C_G\varphi\to\varphi^j$. As $M^k,w\models C_G\varphi$, it follows that $M^k,w\models\varphi^j$. The generated model is fixed from stage $k$ onward, so $M^n,w\models\varphi^j$ for every $n\geq k$, and indeed $M^\alpha,w\models\varphi^j$ for every ordinal $\alpha\geq k$. Thus $S_{ij}^{<\omega}(\varphi)$ implies both $SE_{ij}(\varphi)$ and $SE_{ij}^{Ord}(\varphi)$.

Combining this with the general implications $SE_{ij}\Rightarrow E_{ij}\Rightarrow E_{ij}^{Ord}\Rightarrow S_{ij}^{<\omega}$, with $SE_{ij}^{Ord}\Rightarrow E_{ij}^{Ord}$, and with $S_{ij}^{Ord}\Leftrightarrow SE_{ij}^{Ord}$ (see Lemma~\ref{lem:ordinal-eventual-facts}), we obtain
\begin{equation*}
SE_{ij}\Longleftrightarrow E_{ij}
\Longleftrightarrow SE_{ij}^{Ord}
\Longleftrightarrow E_{ij}^{Ord}
\Longleftrightarrow S_{ij}^{Ord}
\Longleftrightarrow S_{ij}^{<\omega}.
\end{equation*}
\end{proof}
\begin{lemma}\label{lem:two-agent-reversal}
    If $|G|\geq2$, then $S_{10}^{<\omega}\not\Rightarrow E_{10}$ and $S_{01}^{<\omega}\not\Rightarrow E_{01}$.
\end{lemma}
\begin{proof}
We construct two basic formulas $\theta_{10}$ and $\theta_{01}$ such that
$\theta_{10}\in S_{10}^{<\omega}\setminus E_{10}$ and
$\theta_{01}\in S_{01}^{<\omega}\setminus E_{01}$.

Choose distinct agents $a,b\in G$ and atoms $r,s$. For any formula $\chi$, define
\begin{equation*}
L\chi:=
(s\land\Diamond_a(\neg s\land\chi))
\lor
(\neg s\land\Diamond_b(s\land\chi)),
\qquad
D:=L\top.
\end{equation*}
Put $\theta_{10}:=D\land L\neg D$. Also put $\psi:=\neg r\land D$ and
$\theta_{01}:=\psi\lor(\Box_a\psi\land\Box_b\psi)$.

\medskip\noindent
\textbf{Proof of $S_{10}^{<\omega}(\theta_{10})$ and $S_{01}^{<\omega}(\theta_{01})$}
We first prove $S_{10}^{<\omega}(\theta_{10})$ and $S_{01}^{<\omega}(\theta_{01})$. Take any pointed model $M,x$. Suppose that $M,x\models\theta_{10}$. Then $L\neg D$ holds at $x$, so there is an $a$- or $b$-successor $y$ such that $M,y\models\lnot D$. Hence $M,y\models\lnot\theta_{10}$. Thus
$\models\theta_{10}\to\neg C_G\theta_{10}$, or equivalently,
$\models C_G\theta_{10}\to\neg\theta_{10}$. By Lemma~\ref{lem:finite-eventual-facts}(2),
$S_{10}^{<\omega}(\theta_{10})$.

For $\theta_{01}$, suppose that $M,x\models\lnot\theta_{01}$. Then $M,x\models\lnot\psi$ and $M,x\models\lnot\Box_i\psi$ for some $i\in\{a,b\}$. Choose $y$ such that $xR_i y$ and $M,y\models\lnot\psi$. Since $R_i$ is transitive and Euclidean, $xR_i y$ implies $yR_i y$. Therefore $M,y\models\lnot\Box_i\psi$, and hence $M,y\models\lnot\theta_{01}$. Thus $\models\neg\theta_{01}\to\neg C_G\theta_{01}$, or equivalently,
$\models C_G\theta_{01}\to\theta_{01}$. Again by Lemma~\ref{lem:finite-eventual-facts}(2),
$S_{01}^{<\omega}(\theta_{01})$.

\medskip\noindent
\textbf{Proof of $\theta_{10}\notin E_{10}$ and $\theta_{01}\notin E_{01}$}

It remains to show that $\theta_{10}\notin E_{10}$ and $\theta_{01}\notin E_{01}$. We use one common model. Let $W$ consist of the empty sequence $\epsilon$ together with all nonempty finite strictly decreasing sequences of natural numbers. Define the rank by $\rho(\epsilon)=\omega$ and, for a nonempty sequence $\sigma$, let $\rho(\sigma)$ be its last entry. For each $\sigma$, let
\begin{equation*}
C_\sigma:=\{\sigma^\frown q:q<\rho(\sigma)\}
\end{equation*}
where $\sigma^\frown q$ denotes the sequence obtained by appending $q$ to $\sigma$. Thus $C_\sigma$ is the set of children of $\sigma$. For example, $C_\epsilon=\{\langle 0\rangle, \langle 1\rangle,\ldots\}$ and $C_{\langle 2\rangle}=\{\langle 2,0\rangle,\langle 2,1\rangle\}$. Let $s$ hold exactly at sequences of even length and let $r$ hold only at $\epsilon$. Define
\begin{equation*}
R_a:=
\bigcup_{|\sigma|\text{ even}}
(\{\sigma\}\cup C_\sigma)\times C_\sigma,
\qquad
R_b:=
\bigcup_{|\sigma|\text{ odd}}
(\{\sigma\}\cup C_\sigma)\times C_\sigma,
\end{equation*}
and interpret every other agent relation as empty. For each agent, these relations are disjoint unions of blocks of the form $X\times C$ with $C\subseteq X$. Each such block is transitive and Euclidean, so the resulting model is K45 (see Figure~\ref{fig:rank_peeling_model}).

In this model, $L\chi$ simply says that some child one level below satisfies $\chi$. In particular, $D=L\top$ holds exactly when at least one child remains accessible. Thus $D$ detects whether the current block still has a surviving target. 

For a sequence $\sigma$, write $\iota(\sigma)=a$ if $|\sigma|$ is even and $\iota(\sigma)=b$ if $|\sigma|$ is odd. Also put
$C_\sigma^n:=\{\sigma^\frown q:n\leq q<\rho(\sigma)\}$.

We first analyze iteration by $\theta_{10}$, writing $M_{10}^n:=M|^n\theta_{10}$. We claim that, for every $n<\omega$, every $\sigma$, and every $x\in\{\sigma\}\cup C_\sigma$,
\begin{equation*}
R_{\iota(\sigma)}^{M_{10}^n}(x)=C_\sigma^n.
\end{equation*}
At the same time, for every finite-rank $\sigma$,
\begin{equation*}
M_{10}^n,\sigma\models D
\quad\Longleftrightarrow\quad
M_{10}^n,\sigma\models\theta_{10}
\quad\Longleftrightarrow\quad
n<\rho(\sigma).
\end{equation*}
The statement $R_{\iota(\sigma)}^{M_{10}^n}(x)=C_\sigma^n$ means that, at each finite stage, every block is peeled from left to right according to rank: after $n$ announcements, precisely the targets of rank at least $n$ remain accessible. These statements follow simultaneously by induction on $n$. The relation statement is immediate for $n=0$. Suppose it holds at stage $n$. Since the $\iota(\sigma)$-successors of $\sigma$ are precisely the members of $C_\sigma^n$, $D$ holds at $\sigma$ exactly when $n<\rho(\sigma)$. If $n<\rho(\sigma)$, the child $\sigma^\frown n$ is still accessible. Its rank is $n$, so it has no child of rank at least $n$ and hence falsifies $D$ at stage $n$. Thus it witnesses $L\neg D$, and $\theta_{10}$ holds at $\sigma$. Conversely, if $n\geq\rho(\sigma)$, then $D$, and therefore $\theta_{10}$, is false at $\sigma$. Consequently a child $\sigma^\frown q$ survives as a target from stage $n$ to stage $n+1$ exactly when $n<q$, giving $C_\sigma^{n+1}$. This completes the induction.

Now consider the root $\epsilon$. Since $\rho(\epsilon)=\omega$, at every finite stage $n$ the child $\langle n\rangle$ is still accessible and witnesses $L\neg D$. Hence
\begin{equation*}
M_{10}^n,\epsilon\models\theta_{10}
\qquad\text{for every }n<\omega.
\end{equation*}
In particular, $\theta_{10}$ is initially true at $\epsilon$ but never becomes false at any positive finite stage. Therefore $\theta_{10}\notin E_{10}$. Notice also that every arrow target has finite rank, so all arrows disappear at stage $\omega$; the root trajectory is
\begin{equation*}
1,1,1,\ldots,0_\omega,0_{\omega+1},\ldots.
\end{equation*}

We next analyze iteration by $\theta_{01}$, writing $M_{01}^n:=M|^n\theta_{01}$. The key observation is that $\theta_{01}$ and $\psi$ agree at every current arrow target. Indeed, if $y$ is a target of an $i$-arrow and $\psi$ is false at $y$, then K45 gives $yR_i y$, so $\Box_i\psi$ is false at $y$ and hence $\theta_{01}$ is false there. The converse is immediate because $\psi$ is a disjunct of $\theta_{01}$. Thus, at every current arrow target,
$\theta_{01}$ holds iff $\psi$ holds.

No arrow in our model targets the root, so $r$ is false at every arrow target. Hence, at every arrow target, $\psi$ holds iff $D$ holds. It follows by induction, exactly as above, that
\begin{equation*}
R_{\iota(\sigma)}^{M_{01}^n}(x)=C_\sigma^n
\end{equation*}
for every $n<\omega$, every $\sigma$, and every $x\in\{\sigma\}\cup C_\sigma$.

At the root $\epsilon$, $\psi$ is false at every stage because $r$ is true there. At every finite stage $n$, the root still has the $a$-successor $\langle n\rangle$. This state has rank $n$, so $D$, and therefore $\psi$, is false there at stage $n$. Hence $\Box_a\psi$ is false at the root and
\begin{equation*}
M_{01}^n,\epsilon\models\neg\theta_{01}
\qquad\text{for every }n<\omega.
\end{equation*}
Thus $\theta_{01}$ is initially false at the root but never becomes true at any positive finite stage. Therefore $\theta_{01}\notin E_{01}$.

At stage $\omega$, all arrows have disappeared. Since $\psi$ is still false at the root while both $\Box_a\psi$ and $\Box_b\psi$ are vacuously true, $\theta_{01}$ becomes true there. Its root trajectory is therefore
\begin{equation*}
0,0,0,\ldots,1_\omega,1_{\omega+1},\ldots.
\end{equation*}

We have thus constructed
$\theta_{10}\in S_{10}^{<\omega}\setminus E_{10}$ and
$\theta_{01}\in S_{01}^{<\omega}\setminus E_{01}$. Hence
$S_{10}^{<\omega}\not\Rightarrow E_{10}$ and
$S_{01}^{<\omega}\not\Rightarrow E_{01}$ for every $|G|\geq2$.
\end{proof}
\begin{figure}[htbp]
    \centering
    \resizebox{0.98\textwidth}{!}{%
    \begin{tikzpicture}[
        x=1.15cm, y=1.15cm,
        >=Stealth,
        line cap=round,
        line join=round,
        state/.style={circle, draw, minimum size=8mm, inner sep=1pt, font=\small},
        aedge/.style={->, blue, thick},
        bedge/.style={->, red, thick},
        asib/.style={blue, thick},
        bsib/.style={red, thick},
        every node/.style={font=\small}
    ]

    \node[state] (eps) at (0,0) {$\epsilon$};
    \node at (0.9,0.45) {$\rho(\epsilon)=\omega$};
    \node at (-0.35,0.35) {$r$};

    \node[state] (n0) at (-5.5,-2.2) {$\langle 0\rangle$};
    \node[state] (n1) at (-2.8,-2.2) {$\langle 1\rangle$};
    \node[state] (n2) at (0,-2.2) {$\langle 2\rangle$};
    \node[state] (n3) at (3,-2.2) {$\langle 3\rangle$};
    \node[state] (n4) at (6,-2.2) {$\langle 4\rangle$};
    \node at (8.2,-2.2) {$\cdots$};

    \node[state] (n10) at (-2.8,-4.6) {$\langle 1,0\rangle$};

    \node[state] (n20) at (-0.7,-4.6) {$\langle 2,0\rangle$};
    \node[state] (n21) at (0.9,-4.6) {$\langle 2,1\rangle$};

    \node[state] (n30) at (2.0,-4.6) {$\langle 3,0\rangle$};
    \node[state] (n31) at (3.4,-4.6) {$\langle 3,1\rangle$};
    \node[state] (n32) at (4.8,-4.6) {$\langle 3,2\rangle$};

    \node[state] (n40) at (5.8,-4.6) {$\langle 4,0\rangle$};
    \node[state] (n41) at (7.0,-4.6) {$\langle 4,1\rangle$};
    \node[state] (n42) at (8.2,-4.6) {$\langle 4,2\rangle$};
    \node[state] (n43) at (9.6,-4.6) {$\langle 4,3\rangle$};

    \node[state] (n210) at (0.9,-7.0) {$\langle 2,1,0\rangle$};

    \node[state] (n310) at (3.2,-7.0) {$\langle 3,1,0\rangle$};

    \node[state] (n320) at (4.8,-7.0) {$\langle 3,2,0\rangle$};
    \node[state] (n321) at (6.3,-7.0) {$\langle 3,2,1\rangle$};

    \node[state] (n3210) at (6.3,-9.4) {$\langle 3,2,1,0\rangle$};

    \draw[aedge] (eps) -- (n0);
    \draw[aedge] (eps) -- (n1);
    \draw[aedge] (eps) -- (n2);
    \draw[aedge] (eps) -- (n3);
    \draw[aedge] (eps) -- (n4);

    \draw[bedge] (n1) -- (n10);

    \draw[bedge] (n2) -- (n20);
    \draw[bedge] (n2) -- (n21);

    \draw[bedge] (n3) -- (n30);
    \draw[bedge] (n3) -- (n31);
    \draw[bedge] (n3) -- (n32);

    \draw[bedge] (n4) -- (n40);
    \draw[bedge] (n4) -- (n41);
    \draw[bedge] (n4) -- (n42);
    \draw[bedge] (n4) -- (n43);

    \draw[aedge] (n21) -- (n210);

    \draw[aedge] (n31) -- (n310);

    \draw[aedge] (n32) -- (n320);
    \draw[aedge] (n32) -- (n321);

    \draw[bedge] (n321) -- (n3210);

    \draw[asib] (n0) -- (n1) -- (n2) -- (n3) -- (n4);
    \draw[asib] (n4) -- +(1.4,0);

    \draw[bsib] (n20) -- (n21);

    \draw[bsib] (n30) -- (n31) -- (n32);

    \draw[bsib] (n40) -- (n41) -- (n42) -- (n43);

    \draw[asib] (n320) -- (n321);

    \draw[densely dotted] (7.8,-5.55) -- (8.45,-6.55);
    \draw[densely dotted] (8.45,-6.55) -- (9.0,-7.45);

    \node[right] at (10.9,0.0) {$s$};
    \node[right] at (10.9,-2.2) {$\lnot s$};
    \node[right] at (10.9,-4.6) {$s$};
    \node[right] at (10.9,-7.0) {$\lnot s$};

    \end{tikzpicture}%
    }
    \caption{The two-agent rank-peeling model. Blue arrows are $R_a$-arrows and red arrows are $R_b$-arrows. Horizontal colored segments schematically indicate the Euclidean accessibility among targets in the same block. Transitive arrows and reflexive arrows are omitted.}
    \label{fig:rank_peeling_model}
\end{figure}
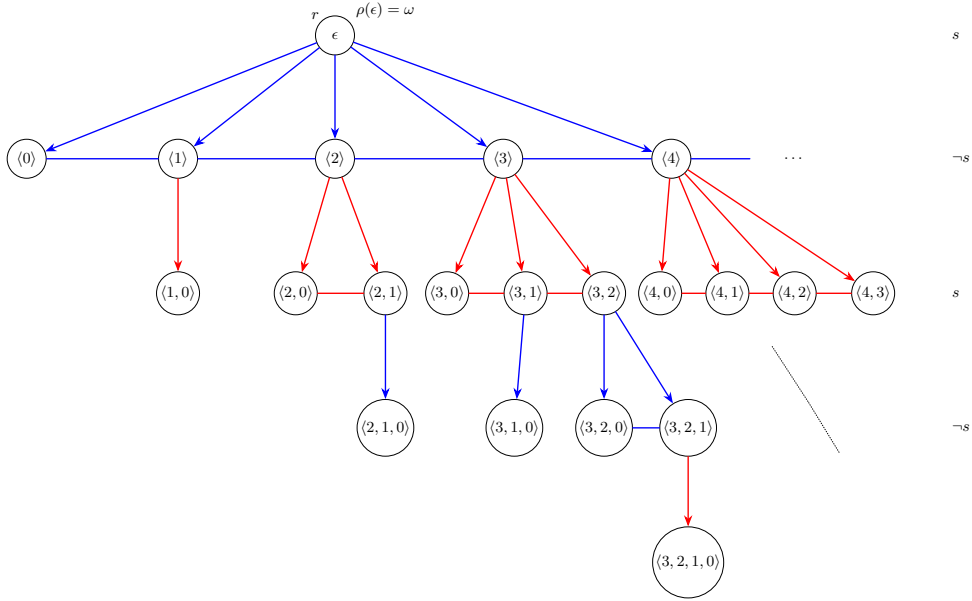
\begin{definition}[Two-peeling model]\label{def:two-peeling-model}
    Choose distinct agents $a,b\in G$. Let
$H:=\{u_{m,0}\mid m\in\omega\}$,
and define the model $M=(W,\{R_i\}_{i\in G},V)$
by
\[
W:=\{w\}\cup
\{u_{m,j}\mid m\in\omega,\ 0\le j\le 2m+1\},
\]
\[
R_a:=
(\{w\}\cup H)\times H
\;\cup\;
\bigcup_{m\in\omega}
\bigcup_{\substack{0\le j\le 2m\\ j\ \mathrm{odd}}}
\{u_{m,j},u_{m,j+1}\}\times\{u_{m,j+1}\},
\]
\[
R_b:=
\bigcup_{m\in\omega}
\bigcup_{\substack{0\le j\le 2m\\ j\ \mathrm{even}}}
\{u_{m,j},u_{m,j+1}\}\times\{u_{m,j+1}\},
\]
and $R_c:=\varnothing$
for every $c\in G\setminus\{a,b\}$.
The valuation is given by
\[
V(r):=\{w\},
\]
\[
V(s):=
\{u_{m,j}\mid
m\in\omega,\ 0\le j\le 2m+1,\ j\ \mathrm{odd}\},
\]

\[
V(p):=
\{u_{m,0}\mid m\in\omega,\ m\ \mathrm{even}\},
\]
and
\[V(t):=H,
\qquad
V(u):=\{u_{0,0}\}.
\]
For every proposition letter $q\notin\{r,s,p,t,u\}$, let $V(q):=\varnothing$ (See Figure~\ref{fig:oscillation-model-three-updates}(a)). The components of each relation are disjoint and every non-empty component has a common target set, so the model is K45.
\end{definition}
\begin{lemma}\label{lem:two-agent-e11-e01-oscillation}
    If $|G|\geq2$, $E_{11}\setminus SE_{11}$
    is nonempty.
    Moreover, $E_{01}\backslash SE_{01}$ is also nonempty.
\end{lemma}
\begin{proof}
    Let
    \begin{equation*}
        \begin{aligned}
            D&:=(s\land\Diamond_a(\neg r\land\neg s))\lor(\neg s\land\Diamond_b(\neg r\land s)),\\
            A&:=\neg r\land\Box_b(r\lor D),\qquad Z:=\Diamond_a(\neg r\land\neg s\land p\land\neg A),\\
            D^A&:=(s\land\Diamond_a(A\land\neg r\land\neg s))\lor(\neg s\land\Diamond_b(A\land\neg r\land s)),\\
            A^A&:=\neg r\land\Box_b(A\to(r\lor D^A)),\qquad
            Z^+:=\Diamond_a(A\land\neg r\land\neg s\land p\land\neg A^A),\\
            \theta_{11}^E&:=A\lor\bigl(r\land(Z\lor\neg Z^+)\bigr).
        \end{aligned}
    \end{equation*}
    
    \medskip\noindent\textbf{Proof of $\theta_{11}^E\in E_{11}\cap E_{01}$.}

    We first show $\theta_{11}^E\in E_{11}\cap E_{01}$.
         Take any model $M$. We claim that, $\theta_{11}^E$ agrees with $A$ at non-$r$ states in $M$, and that $[\uparrow\theta_{11}^E]Z\leftrightarrow Z^+$ is valid:
    At non-$r$ states, $\theta_{11}^E\leftrightarrow A$ is trivial. We prove $[\uparrow\theta_{11}^E]D\leftrightarrow D^A$, $[\uparrow\theta_{11}^E]A\leftrightarrow A^A$, and finally, $[\uparrow\theta_{11}^E]Z\leftrightarrow Z^+$ using the reduction axioms for BPAL (Lemma~\ref{lem:bpal-reduction-axioms}):
    \begin{align*}
        [\uparrow\theta_{11}^E]D
        &\leftrightarrow [\uparrow A]D\\
        &\leftrightarrow (s\land[\uparrow A]\Diamond_a(\lnot r\land\lnot s))\lor (\lnot s\land[\uparrow A]\Diamond_b(\lnot r\land s))\\
        &\leftrightarrow (s\land\Diamond_a(A\land[\uparrow A](\lnot r\land\lnot s))\lor (\lnot s\land\Diamond_b(A\land[\uparrow A](\lnot r\land s))\\
        &\leftrightarrow (s\land\Diamond_a(A\land\lnot r\land\lnot s))\lor (\lnot s\land\Diamond_b(A\land\lnot r\land s))\\
        &\leftrightarrow D^A,\\
        [\uparrow\theta_{11}^E]A
        &\leftrightarrow[\uparrow\theta_{11}^E](\lnot r\land\Box_b(r\lor D))\\
        &\leftrightarrow \lnot r\land\Box_b(\theta_{11}^E\to [\uparrow \theta_{11}^E](r\lor D))\\
        &\leftrightarrow\lnot r\land \Box_b(A\to (r\lor D^A))\\
        &\leftrightarrow A^A,\\
        [\uparrow\theta_{11}^E]Z&\leftrightarrow[\uparrow\theta_{11}^E]\Diamond_a(\lnot r\land\lnot s\land p\land\lnot A)\\
        &\leftrightarrow \Diamond_a(\theta_{11}^E\land[\uparrow\theta_{11}^E](\lnot r\land\lnot s\land p\land \lnot A))\\
        &\leftrightarrow\Diamond_a(A\land \lnot r\land \lnot s\land p\land \lnot A^A)\\
        &\leftrightarrow Z^+.
    \end{align*}
This proves the claim.

    We now check $\theta_{11}^E\in E_{01}$. If $\theta_{11}^E$ is false at an $r$-state in $M$, then $Z$ is false and $Z^+$ is true, so $\theta_{11}^E$ is true after one update by $[\uparrow\theta_{11}^E]Z\leftrightarrow Z^+$. If $\theta_{11}^E$ is false at a non-$r$ state $x$, then $A=\neg r\land\Box_b(r\lor D)$ is false at $x$, so choose a $b$-successor satisfying $\neg r\land\neg D$. The K45 identity $R_b(x)=R_b(y)$ for $xR_by$ shows that every non-$r$ $b$-successor of $x$ also falsifies $A$. Thus, after an announcement of $\theta_{11}^E$, all arrows from $x$ to such states are removed, and $A$ becomes true. Thus, $\theta_{11}^E$ is a true lie and hence $\theta_{11}^E\in E_{01}$.

    We next check $\theta_{11}^E\in E_{11}$. Take any $w\in M$ and suppose $M,w\models\theta_{11}^E$. If $M^1,w\models \theta_{11}^E$, $\theta_{11}^E$ already satisfies $E_{11}$. If $M^1,w\models\lnot\theta_{11}^E$, we have $M^2,w\models\theta_{11}^E$ since $\theta_{11}^E$ is a true lie as we have just seen. Therefore, $\theta_{11}^E\in E_{11}$.

    \medskip\noindent\textbf{Proof of $\theta_{11}^E\notin SE_{11}$ and $\theta_{11}^E\notin SE_{01}$.}
    
    We next show $\theta_{11}^E\notin SE_{11}$ and $\theta_{11}^E\notin SE_{01}$ using the two-peeling model $M$ in Definition~\ref{def:two-peeling-model}.  

We first check $M,w\models\theta_{11}^E$. Since $u_{0,1}$ has no $a$-successor and satisfies $s\land\neg r$,
we have $M,u_{0,1}\models\neg D$. The unique $b$-successor of
$u_{0,0}$ is $u_{0,1}$, and hence
$M,u_{0,0}\models\neg A$. Moreover,
$M,u_{0,0}\models\neg r\land\neg s\land p$. Since $wR_a u_{0,0}$, it follows that
$M,w\models Z$. As $M,w\models r$, we therefore have $M,w\models r\land Z$ and hence $M,w\models\theta_{11}^{E}$.

Write $M^n:=M|^n\theta_{11}^E$
and, for $0\leq \ell\leq m$, put
    $P_{m,\ell}:=\{u_{m,2\ell},u_{m,2\ell+1}\}$.
We first determine where $A$ is false on each branch.

At stage $0$, the terminal state $u_{m,2m+1}$ satisfies $s$ and has
no $a$-successor. Hence $M,u_{m,2m+1}\models\neg D$.
The two states in $P_{m,m}$ have the common $b$-successor
$u_{m,2m+1}$. Since $r$ is false at every branch state, it follows that
    $M,x\models\neg A$
    for every $x\in P_{m,m}$.

On the other hand, if $\ell<m$, then $u_{m,2\ell+1}$ has the
$a$-successor $u_{m,2\ell+2}$, which satisfies
$\neg r\land\neg s$. Hence $M,u_{m,2\ell+1}\models D$, and therefore every state in $P_{m,\ell}$ satisfies $A$.
\begin{claim*}
 For every $n<\omega$ and every $m\geq n$,
\[
    M^n,x\models\neg A
    \quad\Longleftrightarrow\quad
    x\in P_{m,m-n}
\]
for every state $x$ on branch $m$ (see Figure~\ref{fig:oscillation-model-three-updates}).
\end{claim*}
\begin{claimproof}
The case $n=0$ was proved above. 

$(\Leftarrow)$ Suppose the claim holds at stage $n$. Note that $P_{m,m-n}
    =
    \{u_{m,2(m-n)},u_{m,2(m-n)+1}\}$.
Since $r$ is false at every branch state,
$\theta_{11}^E$ agrees there with $A$. Hence the update from $M^n$ to
$M^{n+1}$ deletes every arrow whose target belongs to
$P_{m,m-n}$.

In particular, if $n<m$, the $a$-arrow from
$u_{m,2(m-n)-1}$ to $u_{m,2(m-n)}$ is deleted. Thus
\[
    M^{n+1},u_{m,2(m-n)-1}\models\neg D.
\]
The two states in
    $P_{m,m-n-1}
    =
    \{u_{m,2(m-n)-2},u_{m,2(m-n)-1}\}$
have the common $b$-successor $u_{m,2(m-n)-1}$, and therefore both
falsify $A$ in $M^{n+1}$.

$(\Rightarrow)$ 
For the converse, the common $b$-successor
$u_{m,2(m-n)+1}$ of the states in $P_{m,m-n}$ is itself an
$A$-false target at stage $n$, so the corresponding $b$-arrows are
deleted. Hence the $\Box_b$-conjunct of $A$ is vacuously true there at
stage $n+1$. All other pairs retain the witnesses that made their
corresponding $D$-formulas true. Thus
$P_{m,m-n-1}$ is the unique $A$-false pair on branch $m$ at stage
$n+1$. This proves the claim by induction.
\end{claimproof}
In particular, the head $u_{m,0}$ is $A$-false exactly at stage $m$.
Since all branch heads are $a$-successors of $w$ initially, and an
$a$-arrow to $u_{m,0}$ is deleted precisely in the update following
stage $m$, we obtain $R_a^{M^n}(w)
    =
    \{u_{m,0}\mid m\geq n\}$.
Among these heads, $u_{n,0}$ is the unique one satisfying $\neg A$ at
stage $n$.

Every branch head satisfies $\neg r\land\neg s$, and $M,u_{m,0}\models p$ iff $m$ is even.
Therefore
\[
    M^n,w\models Z
    \quad\Longleftrightarrow\quad
    M^n,u_{n,0}\models p
    \quad\Longleftrightarrow\quad
    n\text{ is even}.
\]
Moreover, by
$[\uparrow\theta_{11}^E]Z\leftrightarrow Z^+$,
\[
    M^n,w\models Z^+
    \quad\Longleftrightarrow\quad
    M^{n+1},w\models Z
    \quad\Longleftrightarrow\quad
    n\text{ is odd}.
\]
Since $r$ is true at $w$ and $A$ is false there, we have
\[
    M^n,w\models\theta_{11}^E
    \quad\Longleftrightarrow\quad
    M^n,w\models Z\lor\neg Z^+
    \quad\Longleftrightarrow\quad
    n\text{ is even}.
\]
Hence the truth-value sequence of $\theta_{11}^E$ at $w$ is
$1,0,1,0,\ldots$.
Therefore
$\theta_{11}^E\notin SE_{11}$.
Since the pointed model $M^1,w$ starts with value $0$ and has the
subsequent sequence $0,1,0,1,\ldots$, we also have
$\theta_{11}^E\notin SE_{01}$.
\end{proof}
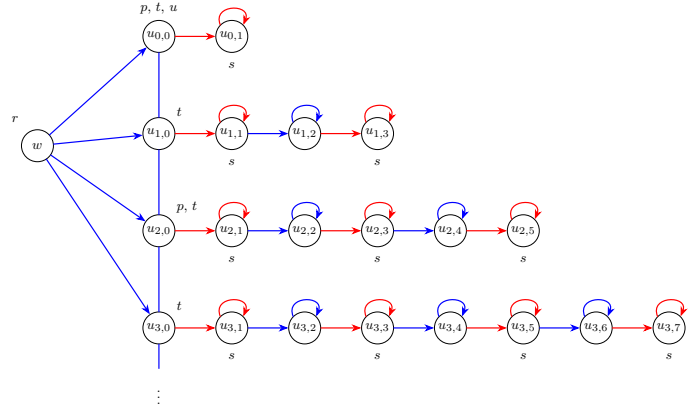
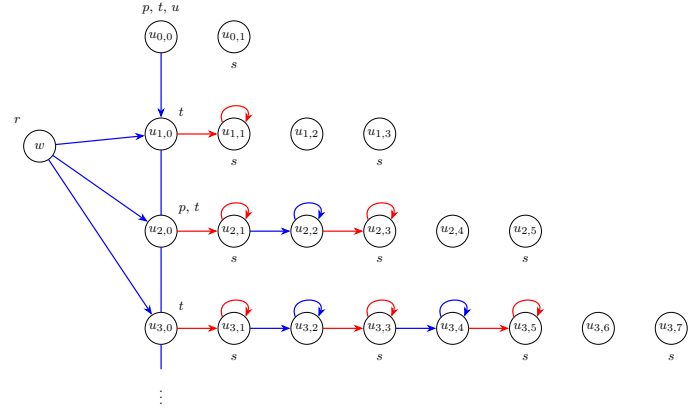
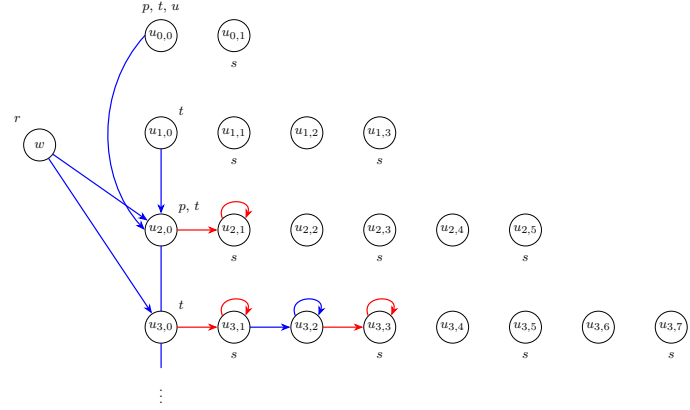
\begin{figure}[p]
    \centering
    \captionsetup[subfigure]{justification=centering}
    \begin{subfigure}{\textwidth}
        \centering
        \resizebox{!}{0.28\textheight}{%
        \begin{tikzpicture}[
            x=0.95cm, y=0.95cm,
            >=Stealth,
            line cap=round,
            line join=round,
            state/.style={circle, draw, minimum size=7.5mm, inner sep=1pt, font=\small},
            aedge/.style={->, blue, thick},
            bedge/.style={->, red, thick},
            asib/.style={blue, thick},
            aloop/.style={->, blue, thick},
            bloop/.style={->, red, thick},
            val/.style={font=\scriptsize},
            every node/.style={font=\small}
        ]

        \node[state] (w) at (0,-2.7) {$w$};
        \node[val] at (-0.55,-2.05) {$r$};

        \node[state] (u00) at (3.0,-0.0) {$u_{0,0}$};
        \node[state] (u10) at (3.0,-2.4) {$u_{1,0}$};
        \node[state] (u20) at (3.0,-4.8) {$u_{2,0}$};
        \node[state] (u30) at (3.0,-7.2) {$u_{3,0}$};
        \node[val] at (3.0, 0.7) {$p$, $t$, $u$};
        \node[val,above right=1pt] at (u10.north east) {$t$};
        \node[val,above right=1pt] at (u20.north east) {$p$, $t$};
        \node[val,above right=1pt] at (u30.north east) {$t$};
        \node at (3.0,-8.8) {$\vdots$};

        \draw[aedge] (w) -- (u00);
        \draw[aedge] (w) -- (u10);
        \draw[aedge] (w) -- (u20);
        \draw[aedge] (w) -- (u30);

        \draw[asib] (u00) -- (u10) -- (u20) -- (u30);
        \draw[asib] (u30) -- +(0,-1.0);

        \node[state] (u01) at (4.8,-0.0) {$u_{0,1}$};
        \node[val] at (4.8,-0.7) {$s$};
        \draw[bedge] (u00) -- (u01);
        \draw[bloop] (u01.135)
            .. controls +(-2mm,5mm) and +(2mm,5mm)
            .. (u01.45);

        \node[state] (u11) at (4.8,-2.4) {$u_{1,1}$};
        \node[state] (u12) at (6.6,-2.4) {$u_{1,2}$};
        \node[state] (u13) at (8.4,-2.4) {$u_{1,3}$};
        \node[val] at (4.8,-3.1) {$s$};
        \node[val] at (8.4,-3.1) {$s$};
        \draw[bedge] (u10) -- (u11);
        \draw[bloop] (u11.135)
            .. controls +(-2mm,5mm) and +(2mm,5mm)
            .. (u11.45);
        \draw[aedge] (u11) -- (u12);
        \draw[aloop] (u12.135)
            .. controls +(-2mm,5mm) and +(2mm,5mm)
            .. (u12.45);
        \draw[bedge] (u12) -- (u13);
        \draw[bloop] (u13.135)
            .. controls +(-2mm,5mm) and +(2mm,5mm)
            .. (u13.45);

        \node[state] (u21) at (4.8,-4.8) {$u_{2,1}$};
        \node[state] (u22) at (6.6,-4.8) {$u_{2,2}$};
        \node[state] (u23) at (8.4,-4.8) {$u_{2,3}$};
        \node[state] (u24) at (10.2,-4.8) {$u_{2,4}$};
        \node[state] (u25) at (12.0,-4.8) {$u_{2,5}$};
        \node[val] at (4.8,-5.5) {$s$};
        \node[val] at (8.4,-5.5) {$s$};
        \node[val] at (12.0,-5.5) {$s$};
        \draw[bedge] (u20) -- (u21);
        \draw[bloop] (u21.135)
            .. controls +(-2mm,5mm) and +(2mm,5mm)
            .. (u21.45);
        \draw[aedge] (u21) -- (u22);
        \draw[aloop] (u22.135)
            .. controls +(-2mm,5mm) and +(2mm,5mm)
            .. (u22.45);
        \draw[bedge] (u22) -- (u23);
        \draw[bloop] (u23.135)
            .. controls +(-2mm,5mm) and +(2mm,5mm)
            .. (u23.45);
        \draw[aedge] (u23) -- (u24);
        \draw[aloop] (u24.135)
            .. controls +(-2mm,5mm) and +(2mm,5mm)
            .. (u24.45);
        \draw[bedge] (u24) -- (u25);
        \draw[bloop] (u25.135)
            .. controls +(-2mm,5mm) and +(2mm,5mm)
            .. (u25.45);

        \node[state] (u31) at (4.8,-7.2) {$u_{3,1}$};
        \node[state] (u32) at (6.6,-7.2) {$u_{3,2}$};
        \node[state] (u33) at (8.4,-7.2) {$u_{3,3}$};
        \node[state] (u34) at (10.2,-7.2) {$u_{3,4}$};
        \node[state] (u35) at (12.0,-7.2) {$u_{3,5}$};
        \node[state] (u36) at (13.8,-7.2) {$u_{3,6}$};
        \node[state] (u37) at (15.6,-7.2) {$u_{3,7}$};
        \node[val] at (4.8,-7.9) {$s$};
        \node[val] at (8.4,-7.9) {$s$};
        \node[val] at (12.0,-7.9) {$s$};
        \node[val] at (15.6,-7.9) {$s$};
        \draw[bedge] (u30) -- (u31);
        \draw[bloop] (u31.135)
            .. controls +(-2mm,5mm) and +(2mm,5mm)
            .. (u31.45);
        \draw[aedge] (u31) -- (u32);
        \draw[aloop] (u32.135)
            .. controls +(-2mm,5mm) and +(2mm,5mm)
            .. (u32.45);
        \draw[bedge] (u32) -- (u33);
        \draw[bloop] (u33.135)
            .. controls +(-2mm,5mm) and +(2mm,5mm)
            .. (u33.45);
        \draw[aedge] (u33) -- (u34);
        \draw[aloop] (u34.135)
            .. controls +(-2mm,5mm) and +(2mm,5mm)
            .. (u34.45);
        \draw[bedge] (u34) -- (u35);
        \draw[bloop] (u35.135)
            .. controls +(-2mm,5mm) and +(2mm,5mm)
            .. (u35.45);
        \draw[aedge] (u35) -- (u36);
        \draw[aloop] (u36.135)
            .. controls +(-2mm,5mm) and +(2mm,5mm)
            .. (u36.45);
        \draw[bedge] (u36) -- (u37);
        \draw[bloop] (u37.135)
            .. controls +(-2mm,5mm) and +(2mm,5mm)
            .. (u37.45);

        \end{tikzpicture}%
        }
        \caption{$M$}
    \end{subfigure}

    \vspace{0.2em}
    \begin{subfigure}{\textwidth}
        \centering
        \resizebox{!}{0.28\textheight}{%
        \begin{tikzpicture}[
            x=0.95cm, y=0.95cm,
            >=Stealth,
            line cap=round,
            line join=round,
            state/.style={circle, draw, minimum size=7.5mm, inner sep=1pt, font=\small},
            aedge/.style={->, blue, thick},
            bedge/.style={->, red, thick},
            asib/.style={blue, thick},
            aloop/.style={->, blue, thick},
            bloop/.style={->, red, thick},
            val/.style={font=\scriptsize},
            every node/.style={font=\small}
        ]

        \node[state] (w) at (0,-2.7) {$w$};
        \node[val] at (-0.55,-2.05) {$r$};

        \node[state] (u00) at (3.0,-0.0) {$u_{0,0}$};
        \node[state] (u10) at (3.0,-2.4) {$u_{1,0}$};
        \node[state] (u20) at (3.0,-4.8) {$u_{2,0}$};
        \node[state] (u30) at (3.0,-7.2) {$u_{3,0}$};
        \node[val] at (3.0, 0.7) {$p$, $t$, $u$};
        \node[val,above right=1pt] at (u10.north east) {$t$};
        \node[val,above right=1pt] at (u20.north east) {$p$, $t$};
        \node[val,above right=1pt] at (u30.north east) {$t$};
        \node at (3.0,-8.8) {$\vdots$};

        \draw[aedge] (w) -- (u10);
        \draw[aedge] (w) -- (u20);
        \draw[aedge] (w) -- (u30);

        \draw[aedge] (u00) -- (u10);
        \draw[asib] (u10) -- (u20) -- (u30);
        \draw[asib] (u30) -- +(0,-1.0);

        \node[state] (u01) at (4.8,-0.0) {$u_{0,1}$};
        \node[val] at (4.8,-0.7) {$s$};

        \node[state] (u11) at (4.8,-2.4) {$u_{1,1}$};
        \node[state] (u12) at (6.6,-2.4) {$u_{1,2}$};
        \node[state] (u13) at (8.4,-2.4) {$u_{1,3}$};
        \node[val] at (4.8,-3.1) {$s$};
        \node[val] at (8.4,-3.1) {$s$};
        \draw[bedge] (u10) -- (u11);
        \draw[bloop] (u11.135)
            .. controls +(-2mm,5mm) and +(2mm,5mm)
            .. (u11.45);

        \node[state] (u21) at (4.8,-4.8) {$u_{2,1}$};
        \node[state] (u22) at (6.6,-4.8) {$u_{2,2}$};
        \node[state] (u23) at (8.4,-4.8) {$u_{2,3}$};
        \node[state] (u24) at (10.2,-4.8) {$u_{2,4}$};
        \node[state] (u25) at (12.0,-4.8) {$u_{2,5}$};
        \node[val] at (4.8,-5.5) {$s$};
        \node[val] at (8.4,-5.5) {$s$};
        \node[val] at (12.0,-5.5) {$s$};
        \draw[bedge] (u20) -- (u21);
        \draw[bloop] (u21.135)
            .. controls +(-2mm,5mm) and +(2mm,5mm)
            .. (u21.45);
        \draw[aedge] (u21) -- (u22);
        \draw[aloop] (u22.135)
            .. controls +(-2mm,5mm) and +(2mm,5mm)
            .. (u22.45);
        \draw[bedge] (u22) -- (u23);
        \draw[bloop] (u23.135)
            .. controls +(-2mm,5mm) and +(2mm,5mm)
            .. (u23.45);

        \node[state] (u31) at (4.8,-7.2) {$u_{3,1}$};
        \node[state] (u32) at (6.6,-7.2) {$u_{3,2}$};
        \node[state] (u33) at (8.4,-7.2) {$u_{3,3}$};
        \node[state] (u34) at (10.2,-7.2) {$u_{3,4}$};
        \node[state] (u35) at (12.0,-7.2) {$u_{3,5}$};
        \node[state] (u36) at (13.8,-7.2) {$u_{3,6}$};
        \node[state] (u37) at (15.6,-7.2) {$u_{3,7}$};
        \node[val] at (4.8,-7.9) {$s$};
        \node[val] at (8.4,-7.9) {$s$};
        \node[val] at (12.0,-7.9) {$s$};
        \node[val] at (15.6,-7.9) {$s$};
        \draw[bedge] (u30) -- (u31);
        \draw[bloop] (u31.135)
            .. controls +(-2mm,5mm) and +(2mm,5mm)
            .. (u31.45);
        \draw[aedge] (u31) -- (u32);
        \draw[aloop] (u32.135)
            .. controls +(-2mm,5mm) and +(2mm,5mm)
            .. (u32.45);
        \draw[bedge] (u32) -- (u33);
        \draw[bloop] (u33.135)
            .. controls +(-2mm,5mm) and +(2mm,5mm)
            .. (u33.45);
        \draw[aedge] (u33) -- (u34);
        \draw[aloop] (u34.135)
            .. controls +(-2mm,5mm) and +(2mm,5mm)
            .. (u34.45);
        \draw[bedge] (u34) -- (u35);
        \draw[bloop] (u35.135)
            .. controls +(-2mm,5mm) and +(2mm,5mm)
            .. (u35.45);

        \end{tikzpicture}%
        }
        \caption{$M^1=M|\theta_{11}^{E}=M|\theta_{11}^{SE}=M|\theta_{11}^{<\omega}$}
    \end{subfigure}

    \vspace{0.2em}
    \begin{subfigure}{\textwidth}
        \centering
        \resizebox{!}{0.28\textheight}{%
        \begin{tikzpicture}[
            x=0.95cm, y=0.95cm,
            >=Stealth,
            line cap=round,
            line join=round,
            state/.style={circle, draw, minimum size=7.5mm, inner sep=1pt, font=\small},
            aedge/.style={->, blue, thick},
            bedge/.style={->, red, thick},
            asib/.style={blue, thick},
            aloop/.style={->, blue, thick},
            bloop/.style={->, red, thick},
            val/.style={font=\scriptsize},
            every node/.style={font=\small}
        ]

        \node[state] (w) at (0,-2.7) {$w$};
        \node[val] at (-0.55,-2.05) {$r$};

        \node[state] (u00) at (3.0,-0.0) {$u_{0,0}$};
        \node[state] (u10) at (3.0,-2.4) {$u_{1,0}$};
        \node[state] (u20) at (3.0,-4.8) {$u_{2,0}$};
        \node[state] (u30) at (3.0,-7.2) {$u_{3,0}$};
        \node[val] at (3.0, 0.7) {$p$, $t$, $u$};
        \node[val,above right=1pt] at (u10.north east) {$t$};
        \node[val,above right=1pt] at (u20.north east) {$p$, $t$};
        \node[val,above right=1pt] at (u30.north east) {$t$};
        \node at (3.0,-8.8) {$\vdots$};

        \draw[aedge] (w) -- (u20);
        \draw[aedge] (w) -- (u30);

        \draw[aedge] (u00.west)
            .. controls +(-1.2,-1.3) and +(-1.2,1.3)
            .. (u20.west);
        \draw[aedge] (u10) -- (u20);
        \draw[asib] (u20) -- (u30);
        \draw[asib] (u30) -- +(0,-1.0);

        \node[state] (u01) at (4.8,-0.0) {$u_{0,1}$};
        \node[val] at (4.8,-0.7) {$s$};

        \node[state] (u11) at (4.8,-2.4) {$u_{1,1}$};
        \node[state] (u12) at (6.6,-2.4) {$u_{1,2}$};
        \node[state] (u13) at (8.4,-2.4) {$u_{1,3}$};
        \node[val] at (4.8,-3.1) {$s$};
        \node[val] at (8.4,-3.1) {$s$};

        \node[state] (u21) at (4.8,-4.8) {$u_{2,1}$};
        \node[state] (u22) at (6.6,-4.8) {$u_{2,2}$};
        \node[state] (u23) at (8.4,-4.8) {$u_{2,3}$};
        \node[state] (u24) at (10.2,-4.8) {$u_{2,4}$};
        \node[state] (u25) at (12.0,-4.8) {$u_{2,5}$};
        \node[val] at (4.8,-5.5) {$s$};
        \node[val] at (8.4,-5.5) {$s$};
        \node[val] at (12.0,-5.5) {$s$};
        \draw[bedge] (u20) -- (u21);
        \draw[bloop] (u21.135)
            .. controls +(-2mm,5mm) and +(2mm,5mm)
            .. (u21.45);

        \node[state] (u31) at (4.8,-7.2) {$u_{3,1}$};
        \node[state] (u32) at (6.6,-7.2) {$u_{3,2}$};
        \node[state] (u33) at (8.4,-7.2) {$u_{3,3}$};
        \node[state] (u34) at (10.2,-7.2) {$u_{3,4}$};
        \node[state] (u35) at (12.0,-7.2) {$u_{3,5}$};
        \node[state] (u36) at (13.8,-7.2) {$u_{3,6}$};
        \node[state] (u37) at (15.6,-7.2) {$u_{3,7}$};
        \node[val] at (4.8,-7.9) {$s$};
        \node[val] at (8.4,-7.9) {$s$};
        \node[val] at (12.0,-7.9) {$s$};
        \node[val] at (15.6,-7.9) {$s$};
        \draw[bedge] (u30) -- (u31);
        \draw[bloop] (u31.135)
            .. controls +(-2mm,5mm) and +(2mm,5mm)
            .. (u31.45);
        \draw[aedge] (u31) -- (u32);
        \draw[aloop] (u32.135)
            .. controls +(-2mm,5mm) and +(2mm,5mm)
            .. (u32.45);
        \draw[bedge] (u32) -- (u33);
        \draw[bloop] (u33.135)
            .. controls +(-2mm,5mm) and +(2mm,5mm)
            .. (u33.45);

        \end{tikzpicture}%
        }
        \caption{$M^2=M|^2\theta_{11}^{E}=M|^2\theta_{11}^{SE}=M|^2\theta_{11}^{<\omega}$}
    \end{subfigure}

    \caption{The two-peeling model in Definition~\ref{def:two-peeling-model}. Blue arrows and red arrows indicate the accessibility relations of $a$ and $b$, respectively. Transitive arrows are omitted.}
    \label{fig:oscillation-model-three-updates}
\end{figure}
\begin{definition}[One-peeling model]\label{def:one-peeling-model}
    Choose distinct agents
$a,b\in G$ and consider the following model
\(
    M=(W,\{R_i\}_{i\in G},V).
\)\footnote{The model is almost the same as the two-peeling model in Definition~\ref{def:two-peeling-model} except that $j\leq m$ now, not $j\leq 2m+1$.}
Let
\(
    H:=\{u_{m,0}\mid m\geq1\}
\)
and put
\(
    W:=\{w\}\cup
    \{u_{m,j}\mid m\geq1,\ 0\leq j\leq m\}.
\)
Define
\(
    R_a:=
    (\{w\}\cup H)\times H
    \;\cup\;
    \bigcup_{m\geq1}
    \bigcup_{\substack{0\leq j<m\\ j\text{ odd}}}
    \{u_{m,j},u_{m,j+1}\}\times\{u_{m,j+1}\},
\)
and
\(
    R_b:=
    \bigcup_{m\geq1}
    \bigcup_{\substack{0\leq j<m\\ j\text{ even}}}
    \{u_{m,j},u_{m,j+1}\}\times\{u_{m,j+1}\}.
\)
For every $c\in G\setminus\{a,b\}$, let $R_c:=\varnothing$.
The valuation is given by
\(
    V(r):=\{w\},
    \qquad
    V(s):=
    \{u_{m,j}\mid m\geq1,\ 0\leq j\leq m,\ j\text{ odd}\},
\)
and
\(
    V(p):=
    \{u_{m,0}\mid m\geq1,\ m\text{ even}\}.
\)
All other proposition letters are false everywhere (see Figure~\ref{fig:e00-oscillation-stage0}).
\end{definition}
\begin{lemma}\label{lem:two-agent-e00-oscillation}
    If $|G|\geq 2$, $E_{00}\backslash SE_{00}$ is nonempty.
\end{lemma}
\begin{proof}
    Put 
    \[L\chi:=(s\land\Diamond_a(\neg r\land\neg s\land\chi))\lor(\neg s\land\Diamond_b(\neg r\land s\land\chi)),\] let $L^k$ denote $k$-fold iteration of $L$, and set $D_k:=L^k\top$, $H_k:=D_k\land\neg D_{k+1}$, and 
    \[X_k:=\Diamond_a(\neg r\land\neg s\land p\land H_k).\]
     Let 
        \[\theta_{00}^E:=(\neg r\land D_1)\lor(r\land X_1\land\neg X_2).\]

        \medskip\noindent\textbf{Proof of $\theta_{00}^E\in E_{00}$} 
    
      For $k\geq 1$, we first show
$\models[\uparrow\theta_{00}^E]D_k\leftrightarrow D_{k+1}$.
At every non-$r$ state, $\theta_{00}^E$ agrees with $D_1$.
Hence, for any formula $\chi$, the reduction axioms for BPAL
(Lemma~\ref{lem:bpal-reduction-axioms}) give
\begin{align*}
    [\uparrow\theta_{00}^E]L\chi
    &\leftrightarrow
    \bigl(
        s\land
        \Diamond_a(
            \theta_{00}^E
            \land
            [\uparrow\theta_{00}^E]
            (\neg r\land\neg s\land\chi)
        )
    \bigr)\\
    &\qquad\lor
    \bigl(
        \neg s\land
        \Diamond_b(
            \theta_{00}^E
            \land
            [\uparrow\theta_{00}^E]
            (\neg r\land s\land\chi)
        )
    \bigr)\\
    &\leftrightarrow
    \bigl(
        s\land
        \Diamond_a(
            D_1\land\neg r\land\neg s
            \land[\uparrow\theta_{00}^E]\chi
        )
    \bigr)\\
    &\qquad\lor
    \bigl(
        \neg s\land
        \Diamond_b(
            D_1\land\neg r\land s
            \land[\uparrow\theta_{00}^E]\chi
        )
    \bigr)\\
    &\leftrightarrow
    L\bigl(D_1\land[\uparrow\theta_{00}^E]\chi\bigr).
\end{align*}
Here the second equivalence uses the fact that every target occurring in
the definition of $L$ satisfies $\neg r$, and therefore
$\theta_{00}^E\leftrightarrow D_1$ at such a target.

We now prove
$\models[\uparrow\theta_{00}^E]D_k\leftrightarrow D_{k+1}$
by induction on $k\geq1$. For $k=1$, since $D_1=L\top$, the preceding
equivalence gives
\[
    [\uparrow\theta_{00}^E]D_1
    \leftrightarrow
    L(D_1\land\top)
    \leftrightarrow
    LD_1
    \leftrightarrow
    D_2.
\]
Suppose
$[\uparrow\theta_{00}^E]D_k\leftrightarrow D_{k+1}$.
Since $D_{k+1}=LD_k$, we obtain
\begin{align*}
    [\uparrow\theta_{00}^E]D_{k+1}
    &\leftrightarrow
    [\uparrow\theta_{00}^E]LD_k\\
    &\leftrightarrow
    L\bigl(
        D_1\land[\uparrow\theta_{00}^E]D_k
    \bigr)\\
    &\leftrightarrow
    L(D_1\land D_{k+1})\\
    &\leftrightarrow
    LD_{k+1}\\
    &\leftrightarrow
    D_{k+2}.
\end{align*}
The penultimate equivalence follows from
$\models D_{k+1}\to D_1$: since \(\models D_k\to\top\), the definition of \(L\) gives \(\models LD_k\to L\top\). Hence \(\models D_{k+1}\to D_1\).

Consequently, the same reduction shifts the other indexed formulas by
one: for every $k\geq1$,
$[\uparrow\theta_{00}^E]H_k\leftrightarrow H_{k+1}$, since
$H_k=D_k\land\neg D_{k+1}$, and $[\uparrow \theta_{00}^E] X_k\leftrightarrow X_{k+1}$ since
\begin{align*}
    [\uparrow\theta_{00}^E]X_k
    &\leftrightarrow
    [\uparrow\theta_{00}^E]
    \Diamond_a(\neg r\land\neg s\land p\land H_k)\\
    &\leftrightarrow
    \Diamond_a\bigl(
        \theta_{00}^E
        \land
        \neg r\land\neg s\land p
        \land
        [\uparrow\theta_{00}^E]H_k
    \bigr)\\
    &\leftrightarrow
    \Diamond_a\bigl(
        D_1\land\neg r\land\neg s\land p\land H_{k+1}
    \bigr)\\
    &\leftrightarrow
    \Diamond_a(\neg r\land\neg s\land p\land H_{k+1})\\
    &\leftrightarrow
    X_{k+1}.
\end{align*}
The third equivalence uses
$\theta_{00}^E\leftrightarrow D_1$ at every $\neg r$-state and
$[\uparrow\theta_{00}^E]H_k\leftrightarrow H_{k+1}$.
The fourth equivalence follows from
$\models H_{k+1}\to D_1$.

Take any $M,w$ and suppose that
$M,w\models\neg\theta_{00}^E$.

Suppose first that $M,w\models\neg r$. Since
$\theta_{00}^E$ agrees with $D_1$ at every non-$r$ state, we have
$M,w\models\neg D_1$. Now $D_1=L\top$ is a disjunction of diamond
formulas, and believed public announcements only delete arrows.
Therefore, once $D_1$ is false at a state, it remains false after every
further update. Since the valuation of $r$ does not change, it follows
that
$M^n,w\models\neg\theta_{00}^E$ for every $n\geq1$.

Suppose next that $M,w\models r$. If
$M^1,w\models\neg\theta_{00}^E$, there is nothing to prove. Otherwise,
suppose
$M^1,w\models\theta_{00}^E$. Since $r$ remains true at $w$, the
definition of $\theta_{00}^E$ gives $M^1,w\models\neg X_2$. The equivalence
$[\uparrow\theta_{00}^E]X_1\leftrightarrow X_2$ is valid in every K45
model, so applying it to $M^1$ yields
\[
    M^2,w\models X_1
    \quad\Longleftrightarrow\quad
    M^1,w\models X_2.
\]
Hence $M^2,w\models\neg X_1$. Since $r$ is still true at $w$, we obtain
$M^2,w\models\neg\theta_{00}^E$.

Thus, whenever $\theta_{00}^E$ is initially false, it is false again
after either one or two updates. Therefore
$\theta_{00}^E\in E_{00}$.

\medskip\noindent\textbf{Proof of $\theta_{00}^E\notin SE_{00}$} 
    
We show $\theta_{00}^E\notin SE_{00}$ using the one-peeling model $M$ in Definition~\ref{def:one-peeling-model}.
We first determine the truth of $D_k$ in the initial model. We claim that
for every $m\geq1$, $0\leq j\leq m$, and $k\geq1$,
\[
    M,u_{m,j}\models D_k
    \quad\Longleftrightarrow\quad
    k\leq m-j.
\]
For $k=1$, this follows directly from the definition of $L$: the
relevant $a$- or $b$-successor of $u_{m,j}$ is $u_{m,j+1}$ when
$j<m$, while no such successor exists when $j=m$. Suppose the claim
holds for $k$. Since $D_{k+1}=LD_k$, we have
\begin{align*}
    M,u_{m,j}\models D_{k+1}
    &\Longleftrightarrow
    j<m\text{ and }M,u_{m,j+1}\models D_k\\
    &\Longleftrightarrow
    k\leq m-j-1\\
    &\Longleftrightarrow
    k+1\leq m-j.
\end{align*}
Thus the claim follows by induction.

In particular, at the head $u_{m,0}$,
\[
    M,u_{m,0}\models H_k
    \quad\Longleftrightarrow\quad
    M,u_{m,0}\models D_k\land\neg D_{k+1}
    \quad\Longleftrightarrow\quad
    m=k.
\]
Since every head is an $a$-successor of $w$, every head satisfies
$\neg r\land\neg s$, and $p$ holds at $u_{m,0}$ exactly when $m$ is
even, it follows that
\[
    M,w\models X_k
    \quad\Longleftrightarrow\quad
    k\text{ is even}.
\]

Write $M^n:=M|^n\theta_{00}^E$. From the validity
$[\uparrow\theta_{00}^E]X_k\leftrightarrow X_{k+1}$, induction on $n$
gives
\[
    M^n,w\models X_k
    \quad\Longleftrightarrow\quad
    M,w\models X_{k+n}.
\]
Hence
\[
    M^n,w\models X_1
    \quad\Longleftrightarrow\quad
    n\text{ is odd},
    \qquad
    M^n,w\models X_2
    \quad\Longleftrightarrow\quad
    n\text{ is even}.
\]
The valuation of $r$ is unchanged by the updates, so
$M^n,w\models r$ for every $n<\omega$. Therefore, by the definition of
$\theta_{00}^E$,
\begin{align*}
    M^n,w\models\theta_{00}^E
    &\Longleftrightarrow
    M^n,w\models X_1\land\neg X_2\\
    &\Longleftrightarrow
    n\text{ is odd}.
\end{align*}
Thus the truth-value sequence of $\theta_{00}^E$ at $w$ is
$0,1,0,1,\ldots$. In particular,
$M,w\models\neg\theta_{00}^E$, but $\theta_{00}^E$ is true at
arbitrarily large finite stages. Hence
$\theta_{00}^E\notin SE_{00}$.
\end{proof}

\begin{figure}[p]
    \centering
    \captionsetup{font=small,skip=4pt}
    \captionsetup[subfigure]{font=small,skip=2pt,justification=centering}

    \newcommand{\EZeroZeroDiagram}[1]{%
        \begin{tikzpicture}[
            x=1cm, y=1cm,
            >=Stealth,
            line cap=round,
            line join=round,
            state/.style={circle,draw,minimum size=7.5mm,
                          inner sep=1pt,font=\small},
            aedge/.style={->,blue,thick},
            bedge/.style={->,red,thick},
            ahead/.style={<->,blue,thick},
            val/.style={font=\scriptsize,inner sep=1pt}
        ]
        \path[use as bounding box] (-0.6,-5.45) rectangle (12.1,0.8);

        \node[state] (w) at (0,-2.175) {$w$};
        \node[val,above=2pt] at (w.north) {$r$};

        \foreach \m in {1,...,4} {
            \foreach \j in {0,...,\m} {
                \node[state] (u\m-\j)
                    at ({2.8+2.15*\j},{-1.45*(\m-1)})
                    {$u_{\m,\j}$};
                \ifodd\j
                    \node[val,below=2pt] at (u\m-\j.south) {$s$};
                \fi
            }
            \ifodd\m\else
                \node[val,above right=1pt]
                    at (u\m-0.north east) {$p$};
            \fi
        }

        \foreach \m in {1,...,4} {
            \ifnum\m<#1\relax\else
                \draw[aedge] (w) -- (u\m-0);
            \fi
        }

        \foreach \m in {1,...,3} {
            \pgfmathtruncatemacro{\nextm}{\m+1}
            \ifnum\m<#1\relax
                \draw[aedge] (u\m-0) -- (u\nextm-0);
            \else
                \draw[ahead] (u\m-0) -- (u\nextm-0);
            \fi
        }
        \draw[blue,thick,densely dotted]
            (u4-0.south) -- (2.8,-4.9);
        \node at (2.8,-5.2) {$\vdots$};

        \foreach \m in {1,...,4} {
            \pgfmathtruncatemacro{\lasttarget}{\m-#1}
            \foreach \j in {1,...,\m} {
                \ifnum\j>\lasttarget\relax\else
                    \pgfmathtruncatemacro{\previousj}{\j-1}
                    \ifodd\j
                        \draw[bedge] (u\m-\previousj) -- (u\m-\j);
                        \draw[bedge] (u\m-\j.135)
                            .. controls +(-2mm,5mm) and +(2mm,5mm)
                            .. (u\m-\j.45);
                    \else
                        \draw[aedge] (u\m-\previousj) -- (u\m-\j);
                        \draw[aedge] (u\m-\j.135)
                            .. controls +(-2mm,5mm) and +(2mm,5mm)
                            .. (u\m-\j.45);
                    \fi
                \fi
            }
        }
        \end{tikzpicture}%
    }

    \begin{subfigure}{\textwidth}
        \centering
        \resizebox{!}{0.28\textheight}{\EZeroZeroDiagram{0}}
        \caption{$M$}
        \label{fig:e00-oscillation-stage0}
    \end{subfigure}\par\vspace{3pt}

    \begin{subfigure}{\textwidth}
        \centering
        \resizebox{!}{0.28\textheight}{\EZeroZeroDiagram{1}}
        \caption{$M^1=M|\theta_{00}^{E}=M|\theta_{11}^{SEO}=M|\theta_{00}^{SEO}$}
        \label{fig:e00-oscillation-stage1}
    \end{subfigure}\par\vspace{3pt}

    \begin{subfigure}{\textwidth}
        \centering
        \resizebox{!}{0.28\textheight}{\EZeroZeroDiagram{2}}
        \caption{$M^2=M|^2\theta_{00}^{E}=M|^2\theta_{11}^{SEO}=M|^2\theta_{00}^{SEO}$}
        \label{fig:e00-oscillation-stage2}
    \end{subfigure}

    \caption{The one-peeling model in Definition~\ref{def:one-peeling-model}. Blue arrows and red arrows indicate the accessibility relations of $a$ and $b$, respectively.}
    \label{fig:e00-oscillation-three-updates}
\end{figure}
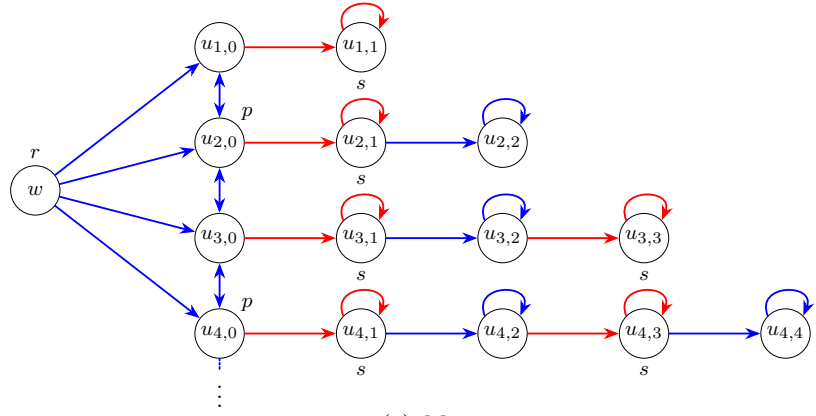
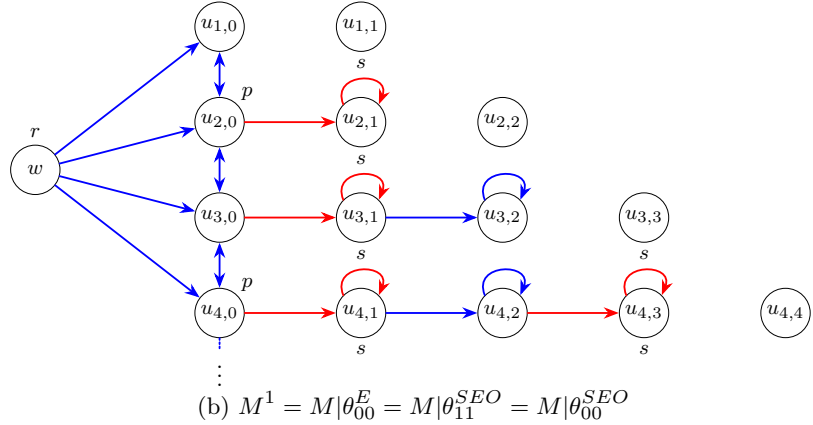
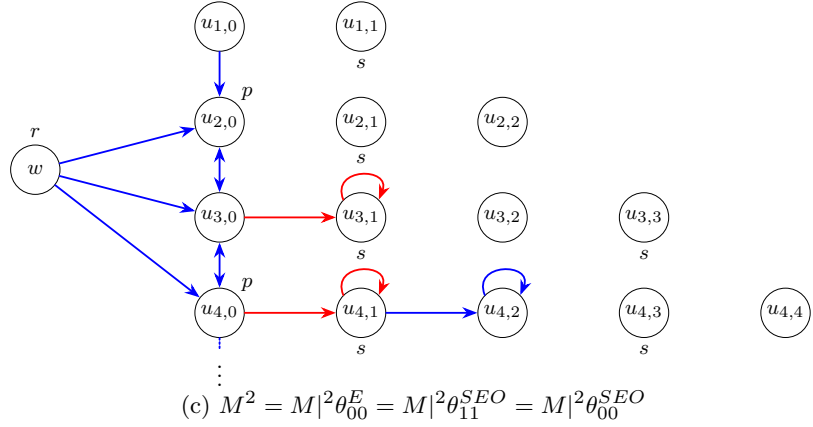

\begin{lemma}\label{lem:two-agent-ordinal-return-separations}
    If $|G|\geq 2$, $SE_{11}^{Ord}\backslash E_{11}$ and $SE_{00}^{Ord}\backslash E_{00}$ are nonempty.
\end{lemma}
\begin{proof}
Choose distinct agents $a,b\in G$ and put
\[
    \chi:=(s\land\Diamond_a\neg s)
        \lor(\neg s\land\Diamond_b s),
    \qquad
    \theta_{11}^{SEO}:=\chi\lor(\Box_a\chi\land\Box_b\chi).
\]

\smallskip\noindent
\textbf{Proof of $\theta_{11}^{SEO}\in SE_{11}^{Ord}$.}
Suppose $M,w\models\theta_{11}^{SEO}$, where $M$ is any K45 model.
Write $M^\alpha:=M|^\alpha\theta_{11}^{SEO}$ and let $R_i^\alpha$
be the accessibility relation for agent $i$ in $M^\alpha$.
By Proposition~\ref{prop:ordinal-stabilization}, choose an ordinal
$\lambda\geq1$ such that $M^\lambda=M^{\lambda+1}$.
Every arrow in $M^\lambda$ has a target satisfying
$\theta_{11}^{SEO}$; otherwise it would be deleted in the next update.

Suppose, towards a contradiction, that
$M^\lambda,w\models\neg\theta_{11}^{SEO}$.
Then $M^\lambda,w\models\neg\chi$ and at least one of
$\Box_a\chi$ and $\Box_b\chi$ is false at $w$.
Thus there are $i\in\{a,b\}$ and a state $v$ such that
$wR_i^\lambda v$ and $M^\lambda,v\models\neg\chi$.
By Euclideanness of $R_i^\lambda$, we have $vR_i^\lambda v$, so $M^\lambda,v\models\neg\Box_i\chi$.
Together with $M^\lambda,v\models\neg\chi$, this gives
$M^\lambda,v\models\neg\theta_{11}^{SEO}$, contradicting the fact
that $v$ is the target of an arrow in $M^\lambda$.
Therefore $M^\lambda,w\models\theta_{11}^{SEO}$.
Since the model is unchanged at every later stage,
$M^\beta,w\models\theta_{11}^{SEO}$ for every $\beta\geq\lambda$.
This proves $\theta_{11}^{SEO}\in SE_{11}^{Ord}$.

\smallskip\noindent
\textbf{Proof of $\theta_{11}^{SEO}\notin E_{11}$.}
Let $m,j$ range over natural numbers and let all the following states
be distinct. We use the one-peeling model in Definition~\ref{def:one-peeling-model}. See Figure~\ref{fig:e00-oscillation-three-updates}. The updated models will also be the same. 

Initially, every head $u_{m,0}$ satisfies $\chi$, since it satisfies
$\neg s$ and has the $b$-successor $u_{m,1}$ satisfying $s$.
Moreover, $R_a(w)=H$ and $R_b(w)=\varnothing$.
Hence $M,w\models\Box_a\chi\land\Box_b\chi$, and therefore
$M,w\models\theta_{11}^{SEO}$.

Write $M^n:=M|^n\theta_{11}^{SEO}$ and let $R_i^n$ be its relations.
We first explain which arrows are deleted at each update.
If $xR_i^n y$ for $i\in\{a,b\}$, Euclideanness gives $yR_i^n y$.
If $\chi$ is false at $y$, then $\Box_i\chi$ is false there as well,
and hence $\theta_{11}^{SEO}$ is false there.
Conversely, $\chi$ implies $\theta_{11}^{SEO}$ by definition.
Thus, at every target of an $a$- or $b$-arrow in $M^n$,
$\chi$ and $\theta_{11}^{SEO}$ have the same truth value.
An existing arrow therefore survives the next update exactly when
its target satisfies $\chi$ in $M^n$.

We claim that, for every $n\geq0$, $m\geq1$, $0\leq j\leq m$,
$x\in W$, and $i\in\{a,b\}$,
\[
    xR_i^n u_{m,j}
    \quad\Longleftrightarrow\quad
    xR_i u_{m,j}\ \text{and}\ n\leq m-j.
\]
The case $n=0$ follows from the definition of $M^0$.
Suppose the claim holds at stage $n$.
At $u_{m,j}$, the definition of $\chi$ selects the $b$-relation when
$j$ is even and the $a$-relation when $j$ is odd.
Initially, this relation has the single successor $u_{m,j+1}$ if
$j<m$, and has no successor if $j=m$.
For $j<m$, the states $u_{m,j}$ and $u_{m,j+1}$ have opposite truth
values of $s$, and the relevant arrow remains at stage $n$ exactly
when $n\leq m-j-1$. Thus
\[
    M^n,u_{m,j}\models\chi
    \quad\Longleftrightarrow\quad
    n<m-j.
\]
Using the condition for an existing arrow to survive, we obtain
\[
\begin{aligned}
    xR_i^{n+1}u_{m,j}
    &\Longleftrightarrow
    xR_i^n u_{m,j}\ \text{and}\ M^n,u_{m,j}\models\chi\\
    &\Longleftrightarrow
    xR_i u_{m,j}\ \text{and}\ n<m-j\\
    &\Longleftrightarrow
    xR_i u_{m,j}\ \text{and}\ n+1\leq m-j.
\end{aligned}
\]
This proves the claim by induction and establishes the displayed
condition for $\chi$ at every finite stage.

In particular, $R_a^n(w)=\{u_{m,0}\mid m\geq1,\ n\leq m\}$ and
$R_b^n(w)=\varnothing$.
Since $s$ is false at $w$, we have $M^n,w\models\neg\chi$ at every
finite stage. For each $n\geq1$, the head $u_{n,0}$ is still an
$a$-successor of $w$, but $M^n,u_{n,0}\models\neg\chi$ because
$n<n$ is false. Hence $M^n,w\models\neg\Box_a\chi$, and therefore
$M^n,w\models\neg\theta_{11}^{SEO}$ for every $n\geq1$.
Since $\theta_{11}^{SEO}$ is true at $M,w$ but false at every positive
finite stage, $\theta_{11}^{SEO}\notin E_{11}$.

Finally, every arrow targets some $u_{m,j}$ and is deleted by stage
$m-j+1$. Thus all relations are empty at stage $\omega$.
Both $\Box_a\chi$ and $\Box_b\chi$ are then true at $w$, and further
updates leave the model unchanged. The truth values at $w$ are
$1,0,0,\ldots,1_\omega,1_{\omega+1},\ldots$.

Put
\[
    D:=(s\land\Diamond_a(\neg r\land\neg s))
      \lor(\neg s\land\Diamond_b(\neg r\land s)),
    \qquad
    A:=\neg r\land D,
\]
and let
\[
    Y:=\Box_a(r\lor A),
    \qquad
    \theta_{00}^{SEO}:=A\lor(r\land\neg Y).
\]

\smallskip\noindent
\textbf{Proof of $\theta_{00}^{SEO}\in SE_{00}^{Ord}$.}
Suppose $M,w\models\neg\theta_{00}^{SEO}$, where $M$ is any K45
model, and write $M^\alpha:=M|^\alpha\theta_{00}^{SEO}$.

Suppose first that $M,w\models\neg r$. At every non-$r$ state,
$\theta_{00}^{SEO}$ agrees with $A$, and $A$ agrees with $D$.
Hence $M,w\models\neg D$. Since believed public announcements only
delete arrows and do not change valuations, neither of the diamond
disjuncts of $D$ can become true at a later stage once it is false.
Therefore
$M^\alpha,w\models\neg D$, and hence
$M^\alpha,w\models\neg\theta_{00}^{SEO}$, for every ordinal
$\alpha$.

Suppose next that $M,w\models r$. By
Proposition~\ref{prop:ordinal-stabilization}, choose an ordinal
$\lambda$ such that
$M^\lambda=M^{\lambda+1}$. Replacing $\lambda$ by $\lambda+1$ if
necessary, we may assume $\lambda>0$.

We claim that
$M^\lambda,w\models\neg\theta_{00}^{SEO}$.
Since $M^\lambda$ is fixed, every arrow remaining in $M^\lambda$ has
a $\theta_{00}^{SEO}$-true target. Thus, if
$wR_a^\lambda v$, then
$M^\lambda,v\models\theta_{00}^{SEO}$.
If $M^\lambda,v\models r$, then
$M^\lambda,v\models r\lor A$ immediately. If instead
$M^\lambda,v\models\neg r$, then
$\theta_{00}^{SEO}$ agrees with $A$ at $v$, so
$M^\lambda,v\models A$, and again
$M^\lambda,v\models r\lor A$.
Consequently every $a$-successor of $w$ satisfies $r\lor A$, and hence
\[
    M^\lambda,w\models Y.
\]
Since $r$ is true at $w$, $A=\neg r\land D$ is false there.
Therefore
$M^\lambda,w\models\neg\theta_{00}^{SEO}$.
As $M^\lambda$ is fixed, this remains true at every later stage.
Thus $\theta_{00}^{SEO}\in SE_{00}^{Ord}$.

\smallskip\noindent
\textbf{Proof of $\theta_{00}^{SEO}\notin E_{00}$.}
For the counterexample, we again use the one-peeling model $M,w$ in Definition~\ref{def:one-peeling-model}. See again Figure~\ref{fig:e00-oscillation-three-updates}. The updated models will also be the same.

No arrow in this model has $w$ as its target. Hence every arrow target
satisfies $\neg r$. At such a state,
\[
    \theta_{00}^{SEO}
    \leftrightarrow A
    \leftrightarrow D.
\]
In the countermodel for
Lemma~\ref{lem:two-agent-e00-oscillation}, the formula
$\theta_{00}^{E}$ also agrees with $D$ at every arrow target.
Therefore iteration by $\theta_{00}^{SEO}$ deletes exactly the same
arrows at every stage as iteration by $\theta_{00}^{E}$.

Initially, every branch head $u_{m,0}$ satisfies $D$, since
$m\geq1$. Hence every $a$-successor of $w$ satisfies $A$, so
$M,w\models Y$. Since $M,w\models r$, we have
\[
    M,w\models\neg\theta_{00}^{SEO}.
\]

Now let $n\geq1$ be finite. As in the countermodel of
Lemma~\ref{lem:two-agent-e00-oscillation}, the head $u_{n,0}$ is still
an $a$-successor of $w$ at stage $n$, but
\[
    M^n,u_{n,0}\models\neg D.
\]
Since $r$ is false at $u_{n,0}$, we also have
$M^n,u_{n,0}\models\neg A$. Thus
$M^n,w\models\neg Y$, and, since $r$ remains true at $w$,
\[
    M^n,w\models\theta_{00}^{SEO}
    \qquad\text{for every finite }n\geq1.
\]

At stage $\omega$, all arrows have been deleted. Hence
$M^\omega,w\models Y$ vacuously, so
$M^\omega,w\models\neg\theta_{00}^{SEO}$. The model is fixed from
that stage onward. Thus the truth-value sequence of
$\theta_{00}^{SEO}$ at $w$ is
$0,1,1,\ldots,0_\omega,0_{\omega+1},\ldots$.

In particular, $\theta_{00}^{SEO}$ is initially false but is true at
every positive finite stage. Therefore
$\theta_{00}^{SEO}\notin E_{00}$.
    \end{proof}

\begin{lemma}\label{lem:two-agent-11-limit-loss-separations}
    If $|G|\geq 2$, $SE_{11}\backslash SE_{11}^{Ord}$ and $S_{11}^{<\omega}\backslash E_{11}^{Ord}$ are nonempty.
\end{lemma}
\begin{proof}
Choose distinct agents $a,b\in G$. Define
\[
\begin{aligned}
    C&:=\neg r\land\neg t,
    \qquad
    P:=\neg r\land t,\\
    D&:=(s\land\Diamond_a(C\land\neg s))
       \lor(\neg s\land\Diamond_b(C\land s)),\\
    A&:=\neg r\land\Box_b(r\lor t\lor D),
    \qquad
    Q:=\Box_a(P\to\neg A),
    \qquad
    B:=A\lor Q,\\
    Y&:=\Diamond_aP,
    \qquad
    Z:=\Diamond_a(P\land u),
    \qquad
    H:=\Box_a(P\to B),\\
    \theta_{11}^{SE}
    &:=(C\land A)\lor(P\land B)\lor(r\land Y),\\
    \theta_{11}^{<\omega}
    &:=(C\land A)\lor(P\land B)
      \lor\bigl(r\land(Z\lor(Y\land H))\bigr).
\end{aligned}
\]

We first show the following two facts that will be used for both $\theta_{11}^{SE}$ and $\theta_{11}^{<\omega}$ ($M^n$ here denotes the updated model either by $\theta_{11}^{SE}$ or $\theta_{11}^{<\omega}$):
\begin{enumerate}
    \item For all $n$, if $M^n,x\models \lnot r\land\lnot A$, then $M^{n+1},x\models A$.
    \item For all $n$, if $M^n,x\models Y$, then $M^{n+1},x\models Y$.
\end{enumerate}

For the first statement, suppose that $M^n, x\models \lnot r\land \lnot A$.
Then there is a $b$-successor $y$ of $x$ such that
$y\models C\land\neg D$. If $z$ is any $C$-valued $b$-successor of
$x$, then transitivity and Euclideanness give
$R_b(x)=R_b(z)$. Since $x$ and $z$ both satisfy $\neg r$, it follows
that $A$ is false at $z$. At a $C$-state, both
$\theta_{11}^{SE}$ and $\theta_{11}^{<\omega}$ agree with $A$.
Hence every $C$-valued $b$-successor of $x$ is a false announcement
target and is deleted by the next update. After that update, every
remaining $b$-successor of $x$ satisfies $r\lor t$, so $M^{n+1},x\models A$. Since later updates only delete arrows, $A$ remains true at $x$
at all later finite stages.

For the second statement, suppose that $M^n,x\models Y$.
Thus $x$ has a $P$-valued $a$-successor. We show that at least one
such successor survives the next update. If some $P$-valued
$a$-successor $y$ satisfies $A$, then $y\models B$, and hence both
displayed announcement formulas are true at $y$. Thus the arrow to
$y$ survives.

Otherwise every $P$-valued $a$-successor of $x$ satisfies $\neg A$.
Choose one such successor $y$. Since $xR_a y$, K45 gives
$R_a(x)=R_a(y)$. Hence every $P$-valued $a$-successor of $y$ also
satisfies $\neg A$, so $y\models Q$. Therefore $y\models B$, and again
both announcement formulas are true at $y$. Thus the arrow to $y$
survives. Consequently, $M^{n+1},x\models Y$.

\medskip\noindent
\textbf{Proof of
$\theta_{11}^{SE}\in SE_{11}$.}

We first show $\theta_{11}^{SE}\in SE_{11}$.
Take any pointed K45 model $M,w$ such that
$M,w\models\theta_{11}^{SE}$, and write
$M^n:=M|^n\theta_{11}^{SE}$.

Suppose first that $M,w\models\neg r$. Since the valuation does not
change, $w$ satisfies either $C$ or $P$ at every stage. Whenever
$A$ is true at $w$, $\theta_{11}^{SE}$ is also true there: if $w$
satisfies $C$, this follows from the disjunct $C\land A$, and if $w$
satisfies $P$, then $A$ implies $B$ and hence $P\land B$.

If $A$ never becomes false at any finite stage, then
$\theta_{11}^{SE}$ remains true at every finite stage. Otherwise,
let $n$ be a finite stage at which $A$ is false. By the first fact
above, $A$ is true from stage $n+1$ onward, and therefore
$\theta_{11}^{SE}$ is also true from stage $n+1$ onward.

Suppose next that $M,w\models r$. Since
$M,w\models\theta_{11}^{SE}$, its third disjunct gives $M,w\models Y$.
By the second fact above, $Y$ remains true at every finite stage.
Since $r$ also remains true, $M^n,w\models\theta_{11}^{SE}$ for every
$n<\omega$. Thus $\theta_{11}^{SE}\in SE_{11}$.

\medskip\noindent\textbf{Proof of $\theta_{11}^{SE}\notin SE_{11}^{Ord}$}
We use the two-peeling model in Definition~\ref{def:two-peeling-model}. See Figure~\ref{fig:oscillation-model-three-updates}.
For $0\leq\ell\leq m$, put
$P_{m,\ell}:=\{u_{m,2\ell},u_{m,2\ell+1}\}$.
Write $M^n:=M|^n\theta_{11}^{SE}$.
We claim that, at every finite stage $n$,

\[
\begin{array}{ll}
\text{(i)}&
\text{for every }m\geq n,\ 
P_{m,m-n}\text{ is the unique $A$-false pair on branch }m,\\[2mm]
\text{(ii)}&
R_a^n(w)=\{u_{m,0}\mid m\geq n\},\\[2mm]
\text{(iii)}&
u_{n,0}\text{ is the unique $B$-false state among these heads.}
\end{array}
\]

At stage $0$, the terminal state $u_{m,2m+1}$ satisfies $C\land s$
and has no $a$-successor satisfying $C\land\neg s$. Hence it falsifies
$D$. Since the two states in $P_{m,m}$ have
$u_{m,2m+1}$ as their common $b$-successor, both falsify $A$.
Every earlier pair is $A$-true, because its common odd
$b$-successor has the next even state as an $a$-successor satisfying
$C\land\neg s$. Thus (i) holds for $n=0$.

All heads are initially $a$-successors of $w$, so (ii) also holds.
The head $u_{0,0}$ is $A$-false, whereas every head $u_{m,0}$ with
$m>0$ is $A$-true. Moreover, $u_{1,0}$ is a $P$-valued
$a$-successor of $u_{0,0}$ satisfying $A$, so
$u_{0,0}\models\neg Q$. Hence $u_{0,0}\models\neg B$, while every
later head satisfies $B$ because it satisfies $A$. This proves (iii)
at stage $0$.

Now suppose (i)--(iii) hold at stage $n$. For every $m>n$, deleting
the arrows whose targets lie in the $A$-false pair
$P_{m,m-n}$ removes the $a$-arrow into its even member. Consequently,
the odd state immediately to its left loses the witness that made
$D$ true, and the preceding pair
$P_{m,m-n-1}$ becomes $A$-false. At the same time, the old
$A$-false pair loses its common $b$-target and becomes $A$-true.
All other pairs remain $A$-true. Hence (i) holds at stage $n+1$.

By (iii), $u_{n,0}$ is the only false announcement target among the
currently accessible heads, so exactly this head is deleted from the
common $a$-successor set. Thus
$R_a^{n+1}(w)=\{u_{m,0}\mid m\geq n+1\}$, proving (ii).
By (i), $u_{n+1,0}$ is now the unique $A$-false surviving head.
The later head $u_{n+2,0}$ is still accessible and satisfies $A$,
so $u_{n+1,0}\models\neg Q$ and therefore
$u_{n+1,0}\models\neg B$. Every later head satisfies $A$, hence $B$.
This proves (iii) and completes the induction.

Since every surviving head satisfies $P$, (ii) gives
$M^n,w\models Y$ for every $n<\omega$. Since $r$ is true at $w$,
\[
    M^n,w\models\theta_{11}^{SE}
    \qquad\text{for every }n<\omega.
\]
On the other hand, every head $u_{m,0}$ is deleted as an
$a$-target after finitely many updates. Hence
$R_a^\omega(w)=\varnothing$, so
$M^\omega,w\models\neg Y$ and therefore
$M^\omega,w\models\neg\theta_{11}^{SE}$.
All remaining arrows are also deleted by stage $\omega$, so the model
is fixed from that stage onward. Thus the truth-value sequence at $w$
is
$1,1,1,\ldots,0_\omega,0_{\omega+1},\ldots$.
Therefore
$\theta_{11}^{SE}\notin SE_{11}^{Ord}$.

\medskip\noindent
\textbf{Proof of
$\theta_{11}^{<\omega}\in
S_{11}^{<\omega}$.}

We first prove
$\theta_{11}^{<\omega}\in S_{11}^{<\omega}$.
Fix $k<\omega$, take any pointed K45 model $M,w$ such that
$M,w\models\theta_{11}^{<\omega}$, and write
$M^n:=M|^n\theta_{11}^{<\omega}$.
Suppose that
$M^k,w\models\neg\theta_{11}^{<\omega}$.
We show that
$M^k,w\models\neg C_G\theta_{11}^{<\omega}$.

Suppose first that $M^k,w\models\neg r$.
Whether $w$ satisfies $C$ or $P$, the falsity of
$\theta_{11}^{<\omega}$ implies $M^k,w\models\neg A$.
Hence there is a $b$-successor $v$ satisfying $C\land\neg D$.
As observed above, K45 gives $R_b(w)=R_b(v)$, so $A$ is also false
at $v$. Since $v$ satisfies $C$,
$M^k,v\models\neg\theta_{11}^{<\omega}$.
Thus a directly accessible state already falsifies the announcement
formula, and hence
$M^k,w\models\neg C_G\theta_{11}^{<\omega}$.

Suppose next that $M^k,w\models r$. Since
$M,w\models\theta_{11}^{<\omega}$ and $r$ is true at $w$ initially,
we have $M,w\models Z\lor(Y\land H)$. Since $Z\to Y$, this implies
$M,w\models Y$. By the persistence of $Y$ proved above,
$M^k,w\models Y$. The assumption
$M^k,w\models\neg\theta_{11}^{<\omega}$ therefore gives
$M^k,w\models\neg Z\land\neg H$.
From $\neg H$ there is an $a$-successor $v$ satisfying
$P\land\neg B$. At a $P$-state,
$\theta_{11}^{<\omega}$ agrees with $B$, so
$M^k,v\models\neg\theta_{11}^{<\omega}$.
Again
$M^k,w\models\neg C_G\theta_{11}^{<\omega}$.

Thus, for every $k<\omega$,
whenever $\theta_{11}^{<\omega}$ is initially true and false at stage
$k$, $C_G\theta_{11}^{<\omega}$ is also false there. Hence
$\theta_{11}^{<\omega}\in S_{11}^{<\omega}$.

\medskip\noindent\textbf{Proof of $\theta_{11}^{<\omega}\notin E_{11}^{Ord}$}

Finally, use the same two-peeling model $M,w$ in Definition~\ref{def:two-peeling-model} and Figure~\ref{fig:oscillation-model-three-updates}.
At every arrow target, $r$ is false, and therefore
$\theta_{11}^{SE}$ and $\theta_{11}^{<\omega}$ both agree with
$(C\land A)\lor(P\land B)$. Hence the two formulas induce exactly the
same sequence of accessibility relations on this model.

Initially, $u_{0,0}$ is an $a$-successor of $w$ satisfying
$P\land u$, so $M,w\models Z$ and therefore
$M,w\models\theta_{11}^{<\omega}$.
At every positive finite stage $n$, the surviving heads are exactly
$u_{m,0}$ with $m\geq n$. Hence $Y$ is true, while $Z$ is false
because $u_{0,0}$ is no longer accessible. Moreover, the accessible
head $u_{n,0}$ is $B$-false, so $H$ is false. Therefore
\[
    M^n,w\models\neg\theta_{11}^{<\omega}
    \qquad\text{for every }1\leq n<\omega.
\]
At stage $\omega$, $Y$ and $Z$ are both false, so
$\theta_{11}^{<\omega}$ is still false; the model is fixed thereafter.
Thus its truth-value sequence at $w$ is
$1,0,0,\ldots,0_\omega,0_{\omega+1},\ldots$.
There is no positive ordinal stage at which the initial truth value
returns. Hence
$\theta_{11}^{<\omega}\notin E_{11}^{Ord}$.
\end{proof}
\begin{lemma}\label{lem:two-agent-00-limit-loss-separations}
    If $|G|\geq 2$, $SE_{00}\backslash SE_{00}^{Ord}$ and $S_{00}^{<\omega}\backslash E_{00}^{Ord}$ are nonempty.
\end{lemma}
\begin{proof}
Choose distinct agents $a,b\in G$. Put
\[
\begin{aligned}
    W&:=\neg r\land\neg t,
    \qquad
    T:=\neg r\land t,\\
    D&:=(s\land\Diamond_a(W\land\neg s))
       \lor(\neg s\land\Diamond_b(W\land s)),\\
    F&:=\Diamond_a(W\land D),
    \qquad
    B:=\Diamond_a(W\land\neg D),\\
    N&:=\Box_a\neg W,
    \qquad
    U:=\Diamond_aT,
    \qquad
    K:=F\lor(N\land U),\\
    \theta_{00}^{SE}
    &:=(W\land D)\lor(T\land K)\lor(r\land N\land U),\\
    \theta_{00}^{<\omega}
    &:=(W\land D)\lor(T\land K)
      \lor\bigl(r\land(B\lor(N\land U))\bigr).
\end{aligned}
\]

At every state satisfying $\neg r$, the two formulas agree: if
$x\models W$, they both agree with $D$ at $x$, while if $x\models T$,
they both agree with $K$ at $x$.

We first establish the following three facts. Fix either
$\varphi=\theta_{00}^{SE}$ or
$\varphi=\theta_{00}^{<\omega}$, and write
$M^n:=M|^n\varphi$.

\begin{enumerate}
    \item If $M^n,x\models W\land\neg D$, then
    $M^m,x\models W\land\neg D$ for every $m\geq n$.

    \item If $M^n,x\models T\land\neg K$, then
    $M^m,x\models T\land\neg K$ for every $m\geq n$.

    \item If $M^{n+1},x\models N\land U$, then
    $M^n,x\models N\land U$.
\end{enumerate}

First, suppose $M^n,x\models W\land\neg D$. Since believed public
announcements only delete arrows and do not change valuations, the
diamond formula $D$ cannot become true at any later stage. Hence
$M^m,x\models W\land\neg D$ for every $m\geq n$. 

Second, suppose $M^n,x\models T\land\neg K$. Then
$M^n,x\models\neg F$. The falsity of $F$ persists under further
updates. If $M^n,x\models\neg U$, then $U$ also remains false, so
$K$ remains false.

Suppose instead that $M^n,x\models U$. Since $K$ is false,
$M^n,x\models\neg N$. Let $y$ be any $T$-valued $a$-successor of
$x$. Since $xR_a^n y$ and $R_a^n$ is transitive and Euclidean,
$R_a^n(x)=R_a^n(y)$. Therefore $F$, $N$, and $U$ have the same truth
values at $x$ and $y$, and hence $M^n,y\models\neg K$. Since
$M^n,y\models T$, we have $M^n,y\models\neg\varphi$. Thus every
$a$-arrow from $x$ to a $T$-valued state is deleted at the next
update. Consequently $U$ is false from stage $n+1$ onward, and since
$F$ remains false, so does $K$.

Third, $N\land U$ cannot become true for the first time at a finite
successor stage. Suppose
$M^{n+1},x\models N\land U$. Choose $y$ such that
$xR_a^{n+1}y$ and $M^{n+1},y\models T$. Since the arrow to $y$
survives the update, $M^n,y\models\varphi$, and therefore
$M^n,y\models K$.

Suppose toward a contradiction that $M^n,x\models\neg N$. Since $xR_a^n y$, we again have
$R_a^n(x)=R_a^n(y)$, so $M^n,y\models\neg N$. As
$M^n,y\models K$, it follows that $M^n,y\models F$. Hence there is
$z$ such that $yR_a^n z$ and $M^n,z\models W\land D$. The equality
$R_a^n(x)=R_a^n(y)$ gives $xR_a^n z$. Since $\varphi$ is true at
every $W\land D$ state, the arrow from $x$ to $z$ survives to stage
$n+1$, contradicting $M^{n+1},x\models N$. Thus
$M^n,x\models N$. Moreover, the surviving $T$-valued successor $y$
already gives $M^n,x\models U$. Hence
$M^{n+1},x\models N\land U$ implies
$M^n,x\models N\land U$.

\medskip\noindent
\textbf{Proof of
$\theta_{00}^{SE}\in SE_{00}$.}

We first show $\theta_{00}^{SE}\in SE_{00}$. Take any pointed K45
model $M,x$ and suppose $M,x\models\neg\theta_{00}^{SE}$.

If $M,x\models W$, then $M,x\models\neg D$, so the first fact implies
$M^n,x\models\neg\theta_{00}^{SE}$ for every $n<\omega$.

If $M,x\models T$, then $M,x\models\neg K$, so the second fact gives
the same conclusion.

Finally, suppose $M,x\models r$. Then
$M,x\models\neg(N\land U)$. If $N\land U$ became true at some
positive finite stage, repeated use of the third fact would imply
that it was already true at stage $0$, a contradiction. Hence
$M^n,x\models\neg\theta_{00}^{SE}$ for every $n<\omega$.
Therefore $\theta_{00}^{SE}\in SE_{00}$.

\medskip\noindent
\textbf{Proof of
$\theta_{00}^{<\omega}\in S_{00}^{<\omega}$.}

Fix $k<\omega$, take any pointed K45 model $M,x$, and suppose
$M,x\models\neg\theta_{00}^{<\omega}$. We show that
$M^k,x\models C_G\theta_{00}^{<\omega}$ implies
$M^k,x\models\neg\theta_{00}^{<\omega}$.

If $M^k,x\models\neg\theta_{00}^{<\omega}$, there is nothing to prove.
Suppose instead that $M^k,x\models\theta_{00}^{<\omega}$. By the
first two facts, an initially false state satisfying $W$ or $T$
cannot become true at a finite stage. Hence $M,x\models r$.

At an $r$-state,
$\theta_{00}^{<\omega}$ agrees with $B\lor(N\land U)$. Since the
formula is initially false, $M,x\models\neg(N\land U)$. By the third
fact, $N\land U$ is still false at stage $k$. Therefore
$M^k,x\models B$. Choose $y$ such that $xR_a^k y$ and
$M^k,y\models W\land\neg D$. At $y$,
$\theta_{00}^{<\omega}$ agrees with $D$, so
$M^k,y\models\neg\theta_{00}^{<\omega}$. Hence
$M^k,x\models\neg C_G\theta_{00}^{<\omega}$.

Thus $S_{00}(k,\theta_{00}^{<\omega})$ holds for every $k<\omega$,
and therefore
$\theta_{00}^{<\omega}\in S_{00}^{<\omega}$.

\medskip\noindent
\textbf{Proof of $\theta_{00}^{SE}\notin SE_{00}^{Ord}$ and $\theta_{00}^{<\omega}\notin E_{00}^{Ord}$.}

We now give one model witnessing both remaining non-membership claims.
Let
\[
    H:=\{u_{m,0}\mid m\geq1\},
    \qquad
    \Omega:=\{w,z\}\cup
    \{u_{m,j}\mid m\geq1,\ 0\leq j\leq m\}.
\]
Define $M=(\Omega,\{R_i\}_{i\in G},V)$ by
\[
\begin{aligned}
    R_a
    &:=
    (\{w,z\}\cup H)\times(\{z\}\cup H)\\
    &\quad\cup
    \bigcup_{m\geq1}
    \bigcup_{\substack{0\leq j<m\\ j\text{ odd}}}
    \{u_{m,j},u_{m,j+1}\}\times\{u_{m,j+1}\},\\
    R_b
    &:=
    \bigcup_{m\geq1}
    \bigcup_{\substack{0\leq j<m\\ j\text{ even}}}
    \{u_{m,j},u_{m,j+1}\}\times\{u_{m,j+1}\}.
\end{aligned}
\]
For every $c\in G\setminus\{a,b\}$, put $R_c:=\varnothing$. Let
\[
    V(r):=\{w\},
    \qquad
    V(t):=\{z\},
    \qquad
    V(s):=
    \{u_{m,j}\mid
      m\geq1,\ 0\leq j\leq m,\ j\text{ odd}\}.
\]
All other proposition letters are false everywhere. Thus
$M,w\models r$, $M,z\models T$, and every branch state satisfies $W$.
The relations are transitive and Euclidean, so $M$ is a K45 model (see Figure~\ref{fig:limitloss00-three-updates}).

No arrow has $w$ as its target. Since $w$ is the only $r$-state,
$\theta_{00}^{SE}$ and $\theta_{00}^{<\omega}$ agree at every arrow
target. Hence they induce the same sequence of accessibility
relations on this model; write this common sequence as $M^n$.

On every branch state the two formulas agree with $D$, and the state
$z$ cannot serve as a $W$-valued witness for either diamond in $D$. Now, we have $M^n,u_{m,j}\models D$ iff $n<m-j$, and at stage $n$ the surviving
branch heads are precisely the $u_{m,0}$ with $m\geq n$.

The state $z$ remains an $a$-target at every finite stage. Indeed, at
stage $n$ the head $u_{n+1,0}$ is still accessible from $z$ and
satisfies $W\land D$. Hence $M^n,z\models F$, so
$M^n,z\models K$, and therefore both announcement formulas are true
at $z$.

It follows that, at every finite stage $n$, the $a$-successors of
$w$ are $z$ together with the surviving branch heads. Hence
$M^n,w\models\neg N$, because a $W$-valued branch head is still
accessible, while $M^n,w\models U$ because $z$ is still accessible.

Therefore
$M^n,w\models\neg\theta_{00}^{SE}$ for every $n<\omega$.
At stage $\omega$ all branch heads have disappeared, while $z$
remains. Thus the only $a$-successor of $w$ is $z$, and
$M^\omega,w\models N\land U$. Hence
$M^\omega,w\models\theta_{00}^{SE}$.

The same is true at $z$: its only remaining $a$-successor is itself,
so $M^\omega,z\models N\land U$, and therefore
$M^\omega,z\models K$. Hence all remaining arrows target a state at
which both announcement formulas are true, and the model is fixed
from stage $\omega$ onward. The truth-value sequence of
$\theta_{00}^{SE}$ at $w$ is
$0,0,0,\ldots,1_\omega,1_{\omega+1},\ldots$.
Thus $\theta_{00}^{SE}\notin SE_{00}^{Ord}$.

Finally consider $\theta_{00}^{<\omega}$ at the same root. At stage
$0$, every branch head satisfies $D$, so
$M,w\models\neg B$. Since also $M,w\models\neg N$, we have
$M,w\models\neg\theta_{00}^{<\omega}$.

For every finite $n\geq1$, the head $u_{n,0}$ is still an
$a$-successor of $w$ and satisfies $\neg D$. Thus
$M^n,w\models B$, and hence
$M^n,w\models\theta_{00}^{<\omega}$. At stage $\omega$,
$M^\omega,w\models N\land U$, so the formula remains true, and the
model is fixed thereafter. Its truth-value sequence at $w$ is
$0,1,1,\ldots,1_\omega,1_{\omega+1},\ldots$.

Thus the initial value $0$ never occurs again at any positive ordinal
stage. Therefore
$\theta_{00}^{<\omega}\notin E_{00}^{Ord}$.
\end{proof}
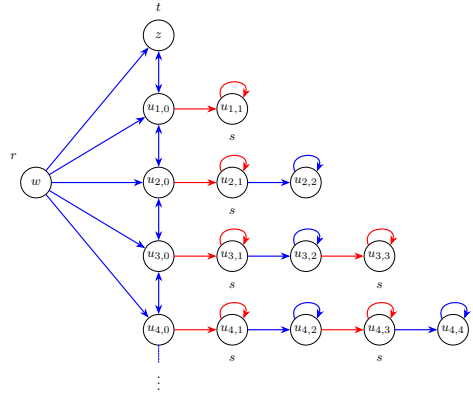
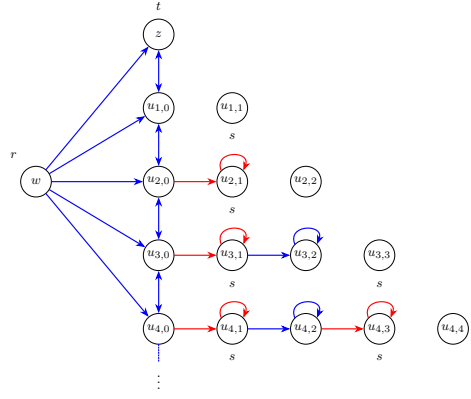
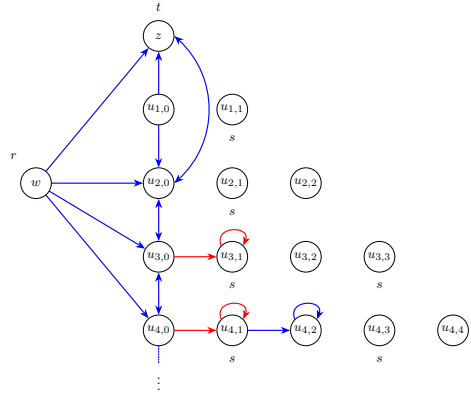
\begin{figure}[p]
    \centering
    \captionsetup{font=small,skip=4pt}
    \captionsetup[subfigure]{font=small,skip=2pt,justification=centering}

    \newcommand{\ZeroZeroLimitDiagram}[1]{%
        \begin{tikzpicture}[
            x=1cm, y=1cm,
            >=Stealth,
            line cap=round,
            line join=round,
            state/.style={circle,draw,minimum size=7.5mm,inner sep=1pt,font=\small},
            aedge/.style={->,blue,thick},
            bedge/.style={->,red,thick},
            ahead/.style={<->,blue,thick},
            aloop/.style={->,blue,thick},
            bloop/.style={->,red,thick},
            val/.style={font=\scriptsize,inner sep=1pt}
        ]
        \path[use as bounding box] (-0.6,-8.9) rectangle (12.5,0.8);

        \node[state] (w) at (0,-3.6) {$w$};
        \node[val] at (-0.55,-2.95) {$r$};
        \node[state] (z) at (3.0,0.0) {$z$};
        \node[val] at (3.0,0.7) {$t$};

        \foreach \m in {1,...,4} {
            \node[state] (u\m-0) at (3.0,{-1.8*\m}) {$u_{\m,0}$};
            \foreach \j in {1,...,4} {
                \ifnum\j>\m\relax\else
                    \node[state] (u\m-\j) at ({3.0+1.8*\j},{-1.8*\m}) {$u_{\m,\j}$};
                    \ifodd\j
                        \node[val] at ({3.0+1.8*\j},{-1.8*\m-0.7}) {$s$};
                    \fi
                \fi
            }
        }
        \node at (3.0,-8.4) {$\vdots$};

        \draw[aedge] (w) -- (z);
        \foreach \m in {1,...,4} {
            \ifnum\m<#1\relax\else
                \draw[aedge] (w) -- (u\m-0);
            \fi
        }

        \def\first{0}
        \foreach \m in {1,...,4} {
            \ifnum\m<#1\relax\else
                \ifnum\first=0\relax
                    \ifnum#1=2\relax
                        \draw[ahead] (z.east)
                            .. controls +(1.1,-1.0) and +(1.1,1.0)
                            .. (u\m-0.east);
                    \else
                        \draw[ahead] (z) -- (u\m-0);
                    \fi
                    \xdef\first{\m}
                \else
                    \draw[ahead] (u\the\numexpr\m-1\relax-0) -- (u\m-0);
                \fi
            \fi
        }
        \draw[blue,thick,densely dotted] (u4-0.south) -- (3.0,-8.0);

        \ifnum#1=2\relax
            \draw[aedge] (u1-0) -- (z);
            \draw[aedge] (u1-0) -- (u2-0);
        \fi

        \foreach \m in {1,...,4} {
            \foreach \j in {1,...,4} {
                \ifnum\j>\m\relax\else
                    \ifnum#1>\numexpr\m-\j\relax
                    \else
                        \pgfmathtruncatemacro{\prev}{\j-1}
                        \ifodd\j
                            \draw[bedge] (u\m-\prev) -- (u\m-\j);
                            \draw[bloop] (u\m-\j.135)
                                .. controls +(-2mm,5mm) and +(2mm,5mm)
                                .. (u\m-\j.45);
                        \else
                            \draw[aedge] (u\m-\prev) -- (u\m-\j);
                            \draw[aloop] (u\m-\j.135)
                                .. controls +(-2mm,5mm) and +(2mm,5mm)
                                .. (u\m-\j.45);
                        \fi
                    \fi
                \fi
            }
        }
        \end{tikzpicture}%
    }

    \begin{subfigure}{\textwidth}
        \centering
        \resizebox{!}{0.27\textheight}{\ZeroZeroLimitDiagram{0}}
        \caption{$M$}
        \label{fig:limitloss00-stage0}
    \end{subfigure}\par\vspace{2pt}

    \begin{subfigure}{\textwidth}
        \centering
        \resizebox{!}{0.27\textheight}{\ZeroZeroLimitDiagram{1}}
        \caption{$M^1=M|\theta_{00}^{SE}=M|\theta_{00}^{<\omega}$}
        \label{fig:limitloss00-stage1}
    \end{subfigure}\par\vspace{2pt}

    \begin{subfigure}{\textwidth}
        \centering
        \resizebox{!}{0.27\textheight}{\ZeroZeroLimitDiagram{2}}
        \caption{$M^2=M|^2\theta_{00}^{SE}=M|^2\theta_{00}^{<\omega}$}
        \label{fig:limitloss00-stage2}
    \end{subfigure}

    \caption{A finite prefix of the common countermodel for the two separations in Lemma~\ref{lem:two-agent-00-limit-loss-separations}. Blue arrows represent $R_a$ and red arrows represent $R_b$. At each finite stage, the accessibility relations induced by $\theta_{00}^{SE}$ and $\theta_{00}^{<\omega}$ coincide. The blue chain through $z$ and the surviving branch heads schematically indicates common $R_a$-accessibility among those states; nonadjacent arrows and self-loops at $z$ and the branch heads are omitted.}
    \label{fig:limitloss00-three-updates}
\end{figure}
\section{Main results}\label{sec:main-results}
\begin{theorem}\label{thm:finite-classification}
    Let $\varphi\in\mathcal{L}$. In multi-agent K45, we have the following. In single-agent K45, all the conditions in each item are equivalent. Every displayed right arrow is strict when $|G|\geq 2$.
    \begin{enumerate}[label=(\arabic*)]
        \item
              \begin{equation*}
                  \begin{aligned}
                                        & \varphi \text{ is strongly eventually successful }                                                          \\
                                        \Longleftrightarrow{}& \forall M,w\,[\sigma_0^{M,w}=1\Rightarrow\lim_{n\to\infty}\sigma_n^{M,w}=1]\\
                      \Longrightarrow{} & \varphi\text{ is eventually successful}\\                   \Longleftrightarrow{}& \forall M,w\,[\sigma_0^{M,w}=1\Rightarrow\limsup_{n\to\infty}\sigma_n^{M,w}=1]                                                         \\
                      \Longleftrightarrow{} & \exists N\geq 1,\quad\models\varphi\to\bigvee_{n=1}^N [\uparrow\varphi]^n\varphi\\ 
                      \Longrightarrow{} & \forall k\geq 0,\quad \models\varphi\to[\uparrow\varphi]^k(\varphi\leftrightarrow(\varphi\lor C_G\varphi)).
                  \end{aligned}
              \end{equation*}
        \item
              \begin{equation*}
                  \begin{aligned}
                                            & \varphi \text{ is strongly eventually self-refuting }                                                             \\\Longleftrightarrow{}& \forall M,w\,\lim_{n\to\infty}\sigma_n^{M,w}=0       \\
                      \Longleftrightarrow{} & \varphi\text{ is eventually self-refuting}\\               \Longleftrightarrow{}& \forall M,w\,\liminf_{n\to\infty}\sigma_n^{M,w}=0\\
                      \Longleftrightarrow{} & \exists N\geq 1,\quad\models[\uparrow\varphi]^N\left(\lnot\varphi\land\bigwedge_{i\in G}\Box_i\bot\right)         \\
                      \Longrightarrow{} & \varphi \text{ is always informative when true}                                                                   \\
                      \Longleftrightarrow{} & \forall k\geq 0,\quad \models\varphi\to[\uparrow\varphi]^k(\varphi\leftrightarrow(\varphi\land \lnot C_G\varphi)) \\
                      \Longleftrightarrow{} & \models C_G\varphi\to \lnot\varphi                                                                                \\
                      \Longleftrightarrow{} & \models \varphi\leftrightarrow (\varphi\land\lnot C_G\varphi).
                  \end{aligned}
              \end{equation*}
        \item
              \begin{equation*}
                  \begin{aligned}
                                            & \varphi \text{ is a strong eventual true lie }                                                                       \\   \Longleftrightarrow{}& \forall M,w\,\lim_{n\to\infty}\sigma_n^{M,w}=1\\
                      \Longrightarrow{}     & \varphi\text{ is an eventual true lie }                                                                         \\
                         \Longleftrightarrow{}& \forall M,w\,\limsup_{n\to\infty}\sigma_n^{M,w}=1\\
                         \Longleftrightarrow{} & \exists N\geq 1\,\quad\models\bigvee_{n=1}^N [\uparrow\varphi]^n\varphi\\
                      \Longrightarrow{} & \varphi \text{ is always informative when false}                                                                \\
                      \Longleftrightarrow{} & \forall k\geq 0,\quad \models\lnot\varphi\to[\uparrow\varphi]^k(\varphi\leftrightarrow(\varphi\lor C_G\varphi)) \\
                      \Longleftrightarrow{} & \models C_G\varphi\to \varphi                                                                                   \\
                      \Longleftrightarrow{} & \models \varphi\leftrightarrow (\varphi\lor C_G\varphi).
                  \end{aligned}
              \end{equation*}
        \item
              \begin{equation*}
                  \begin{aligned}
                                        & \varphi \text{ is a strong eventual impossible lie }                                                                  \\   \Longleftrightarrow{}& \forall M,w\,[\sigma_0^{M,w}=0\Rightarrow\lim_{n\to\infty}\sigma_n^{M,w}=0]\\
                      \Longrightarrow{} & \varphi\text{ is an eventual impossible lie}                                                                            \\   \Longleftrightarrow{}& \forall M,w\,[\sigma_0^{M,w}=0\Rightarrow\liminf_{n\to\infty}\sigma_n^{M,w}=0]\\
                      \Longleftrightarrow{} & \exists N\geq 1,\quad\models\neg\varphi\to\bigvee_{n=1}^N [\uparrow\varphi]^n\neg\varphi\\
                      \Longrightarrow{} & \forall k\geq 0,\quad \models\lnot\varphi\to[\uparrow\varphi]^k(\varphi\leftrightarrow(\varphi\land \lnot C_G\varphi)).
                  \end{aligned}
              \end{equation*}
    \end{enumerate}
\end{theorem}
\begin{proof}
    The equivalences with the limit conditions for strongly eventual
    success and strong eventual impossible lies follow immediately
    from the definitions, since a sequence with values in $\{0,1\}$
    converges to $1$ or $0$ iff it is eventually constantly $1$ or
    $0$, respectively.

    For strong eventual true lies, suppose first that $\varphi$ is
    a strong eventual true lie. Take any pointed model $M,w$. If
    $\sigma_0^{M,w}=0$, then $\lim_{n\to\infty}\sigma_n^{M,w}=1$ by definition.
    Suppose $\sigma_0^{M,w}=1$. If $\sigma_n^{M,w}=1$ for every $n$, there is
    nothing to prove. Otherwise, choose $m\geq1$ such that
    $\sigma_m^{M,w}=0$. Applying strong eventual truth to the pointed tail
    model $M^m,w$ gives $\lim_{n\to\infty}\sigma_n^{M,w}=1$. Conversely, if
    every truth-value sequence converges to $1$, then in particular
    every sequence with $\sigma_0^{M,w}=0$ does so. Hence
    \[
        \varphi\text{ is a strong eventual true lie}
        \quad\Longleftrightarrow\quad
        \forall M,w\;
        \lim_{n\to\infty}\sigma_n^{M,w}=1.
    \]

    We next verify the three characterizations using $\limsup$ or
    $\liminf$. If $\varphi\in E_{11}$ and $\sigma_0^{M,w}=1$, then every
    occurrence of $1$ is followed by a later occurrence of $1$, by
    applying $E_{11}$ to the corresponding tail model. Hence $1$
    occurs infinitely often and
    $\limsup_{n\to\infty}\sigma_n^{M,w}=1$. Conversely, the latter condition
    immediately gives a positive stage at which $\varphi$ is true.
    Thus
    \[
        E_{11}(\varphi)
        \Longleftrightarrow
        \forall M,w\,
        [\sigma_0^{M,w}=1\Rightarrow
        \limsup_{n\to\infty}\sigma_n^{M,w}=1].
    \]
    The same argument gives
    $E_{00}(\varphi)\Longleftrightarrow
    \forall M,w\,[\sigma_0^{M,w}=0\Rightarrow
    \liminf_{n\to\infty}\sigma_n^{M,w}=0]$.

    For $E_{01}$, suppose first that $E_{01}(\varphi)$ holds and take
    any pointed model $M,w$. If $1$ occurred only finitely often, then
    $\sigma_n^{M,w}=0$ for every sufficiently large $n$. Applying
    $E_{01}(\varphi)$ to a tail model beginning at such a stage would
    give a later occurrence of $1$, a contradiction. Hence
    $\limsup_{n\to\infty}\sigma_n^{M,w}=1$. Conversely, if
    $\limsup_{n\to\infty}\sigma_n^{M,w}=1$ for every pointed model, then
    every pointed model with $\sigma_0^{M,w}=0$ has a positive stage at
    which $\varphi$ is true. Therefore
    \[
        E_{01}(\varphi)
        \Longleftrightarrow
        \forall M,w\;
        \limsup_{n\to\infty}\sigma_n^{M,w}=1.
    \]

    The uniform finite-bound characterizations in (1) and (4) follow
    directly from Lemma~\ref{lem:finite-eventual-facts}(3).
    For (3), suppose $E_{01}(\varphi)$ holds. By the same lemma, there
    is an $N\geq1$ such that
    \[
        \models
        \neg\varphi\to
        \bigvee_{n=1}^{N}[\uparrow\varphi]^n\varphi.
    \]
    Take any pointed model $M,w$. If $M^1,w\models\varphi$, then the
    required disjunction already holds at stage $1$. Otherwise,
    $M^1,w\models\neg\varphi$, so applying the displayed validity to
    the pointed model $M^1,w$ yields
    $M^{n+1},w\models\varphi$ for some $1\leq n\leq N$. Hence
    \[
        \models
        \bigvee_{n=1}^{N+1}[\uparrow\varphi]^n\varphi.
    \]
    The converse follows immediately from the definition of
    $E_{01}(\varphi)$.

    For (2), Lemma~\ref{lem:finite-eventual-facts}(4) gives the
    equivalence of strongly eventual self-refutation, eventual
    self-refutation, and the existence of an $N\geq1$ such that
    \[
        \models
        [\uparrow\varphi]^N
        \left(
            \neg\varphi\land
            \bigwedge_{i\in G}\Box_i\bot
        \right).
    \]
    The last condition implies
    $\lim_{n\to\infty}\sigma_n^{M,w}=0$ at every pointed model, since after
    stage $N$ the model is edgeless and $\varphi$ is false everywhere.
    Conversely, if $\lim_{n\to\infty}\sigma_n^{M,w}=0$ at every pointed
    model, then every initially true instance of $\varphi$ eventually
    becomes false, so $E_{10}(\varphi)$ holds. The limit condition implies the corresponding \(\liminf\) condition. Conversely, if \(\liminf_{n\to\infty}\sigma_n^{M,w}=0\) at every pointed model, then every initially true instance becomes false at some positive finite stage, so \(E_{10}(\varphi)\) holds.

    The implications from strong eventuality to eventuality follow
    from the definitions. Lemma~\ref{lem:finite-eventual-facts}(2)
    gives the last implication in (1), the last implication in (4),
    and the implication from an eventual true lie to the remaining
    conditions in (3). The equivalences involving always
    informativeness follow from the same lemma and the propositional
    equivalences between $C_G\varphi\to\varphi$ and
    $\varphi\leftrightarrow(\varphi\lor C_G\varphi)$, and between
    $C_G\varphi\to\neg\varphi$ and
    $\varphi\leftrightarrow(\varphi\land\neg C_G\varphi)$.

    The single-agent claim follows from
    Lemma~\ref{lem:single-agent_collapse}, together with the
    equivalences established above.

    For $|G|\geq2$, strictness in (1) is witnessed by
    Lemmas~\ref{lem:two-agent-e11-e01-oscillation}
    and~\ref{lem:two-agent-11-limit-loss-separations}.
    Strictness in (2) follows from
    Lemma~\ref{lem:two-agent-reversal}.
    Strictness in (3) follows from the $E_{01}\setminus SE_{01}$
    witness in Lemma~\ref{lem:two-agent-e11-e01-oscillation} and from
    Lemma~\ref{lem:two-agent-reversal}.
    Strictness in (4) is witnessed by
    Lemmas~\ref{lem:two-agent-e00-oscillation}
    and~\ref{lem:two-agent-00-limit-loss-separations}.
    For larger groups, all additional accessibility relations are
    interpreted as empty.
\end{proof}

In all the four cases, strong eventual notions are characterized by limit and eventual notions are characterized by (i) limit superior or limit inferior, and (ii) uniform bounds. Eventual notions further imply the fixed-point views of the sentences $\varphi\land\lnot C_G\varphi$ and $\varphi\lor C_G\varphi$.

We can view the result from three sets of categories: \emph{preservation} (11, 00) vs \emph{reversal} (10, 01), \emph{true-initial} (11, 10) vs \emph{false-initial} (01, 00), and \emph{true-target} (11, 01) vs \emph{false-target} (10, 00).

On the reversal side, the preconditions for limit and limsup/liminf characterizations can be omitted, $S_{ij}^{<\omega}$ can be simplified to $\models C_G\varphi\to\lnot\varphi$ (``common belief of $\varphi$ implies $\lnot\varphi$'') or $\models C_G\varphi\to\varphi$ (``common belief of $\varphi$ implies $\varphi$''), and there are always informativeness characterizations. On the other hand, those are absent on the preservation side.

As for true-initial and false-initial sides, the preservation cases are symmetric in the displayed characterizations, whereas the reversal cases exhibit an additional asymmetry: \(SE_{10}=E_{10}\), but \(SE_{01}\subsetneq E_{01}\) when $|G|\geq 2$. Finally, the true-target side has limsup and $\varphi\lor C_G\varphi$ while the false-target side has liminf and $\varphi\land\lnot C_G\varphi$ in their characterizations.  

Next, we analyze, for each $(i,j)\in \{0,1\}^2$, the logical relationship among $E_{ij}$, $SE_{ij}$, $E_{ij}^{Ord}$, $SE_{ij}^{Ord}$, $S_{ij}^{<\omega}$, and $S_{ij}^{Ord}$. We take the three perspectives: strong eventual vs eventual, finite vs transfinite, and what condition $S_{ij}^{Ord}$ is equivalent to. The incomparability \(E_{ii}\parallel SE_{ii}^{Ord}\) is not needed from these perspectives, but we record it for completeness since it follows immediately from the preceding lemmas.
\begin{theorem}\label{thm:transfinite-classification}
    For formulas in $\mathcal{L}$: In multi-agent K45, we have the following. In single-agent K45, all the conditions in each item are equivalent. Every displayed right arrow is strict when $|G|\geq 2$.
    \begin{enumerate}[label=(\arabic*)]
        \item
              \begin{equation*}
                  SE_{11}\Rightarrow E_{11}\Rightarrow E^{Ord}_{11}
                  \Rightarrow S^{<\omega}_{11},
                  \qquad
                  S_{11}^{Ord}\Leftrightarrow SE_{11}^{Ord}\Rightarrow E_{11}^{Ord}.
              \end{equation*}
              Moreover,
              \begin{equation*}
                  E_{11}\parallel SE^{Ord}_{11},
                  \qquad
                  SE_{11}\parallel SE^{Ord}_{11}.
              \end{equation*}
        \item
            \begin{equation*}
                SE_{10}\Leftrightarrow E_{10}\Rightarrow E_{10}^{Ord}\Leftrightarrow SE_{10}^{Ord} \Leftrightarrow S_{10}^{Ord} \Leftrightarrow S_{10}^{<\omega}
            \end{equation*}
        \item
              \begin{equation*}
                SE_{01}\Rightarrow E_{01}\Rightarrow E_{01}^{Ord}\Leftrightarrow SE_{01}^{Ord} \Leftrightarrow S_{01}^{Ord} \Leftrightarrow S_{01}^{<\omega}
            \end{equation*}
        \item
              \begin{equation*}
                  SE_{00}\Rightarrow E_{00}\Rightarrow E^{Ord}_{00}
                  \Rightarrow S^{<\omega}_{00},
                  \qquad
                  S_{00}^{Ord}\Leftrightarrow SE_{00}^{Ord}\Rightarrow E_{00}^{Ord}.
              \end{equation*}
              Moreover,
              \begin{equation*}
                  E_{00}\parallel SE^{Ord}_{00},
                  \qquad
                  SE_{00}\parallel SE^{Ord}_{00}.
              \end{equation*}

    \end{enumerate}
    Here $\parallel$ denotes incomparability.
\end{theorem}

\begin{proof}
    The valid implications and the equivalence
    $S_{ij}^{Ord}\Leftrightarrow SE_{ij}^{Ord}$ follow from
    Lemma~\ref{lem:ordinal-eventual-facts}. That lemma also gives the
    four-way equivalence when $i\neq j$.
    Lemma~\ref{lem:finite-eventual-facts}(4) gives
    $SE_{10}\Leftrightarrow E_{10}$, while
    $SE_{01}\Rightarrow E_{01}$ follows directly from the definitions.

    In the $11$ case,
Lemma~\ref{lem:two-agent-e11-e01-oscillation} gives the strictness of
$SE_{11}\Rightarrow E_{11}$, while
Lemma~\ref{lem:two-agent-ordinal-return-separations} gives
$\theta_{11}^{SEO}\in SE_{11}^{Ord}\setminus E_{11}$, and hence also
$\theta_{11}^{SEO}\in E_{11}^{Ord}\setminus E_{11}$.
Lemma~\ref{lem:two-agent-11-limit-loss-separations} gives
$\theta_{11}^{SE}\in SE_{11}\setminus SE_{11}^{Ord}$; since
$SE_{11}\subseteq E_{11}\subseteq E_{11}^{Ord}$, this also yields
$\theta_{11}^{SE}\in E_{11}^{Ord}\setminus SE_{11}^{Ord}$.
These witnesses also establish the two stated incomparability results. Lemma~\ref{lem:two-agent-11-limit-loss-separations} gives
$\theta_{11}^{<\omega}\in
S_{11}^{<\omega}\setminus E_{11}^{Ord}$, proving the strictness of
$E_{11}^{Ord}\Rightarrow S_{11}^{<\omega}$.

The $00$ case is analogous:
Lemma~\ref{lem:two-agent-e00-oscillation} gives the strictness of
$SE_{00}\Rightarrow E_{00}$,
Lemma~\ref{lem:two-agent-ordinal-return-separations} gives
$\theta_{00}^{SEO}\in SE_{00}^{Ord}\setminus E_{00}$ and hence
$\theta_{00}^{SEO}\in E_{00}^{Ord}\setminus E_{00}$, and
Lemma~\ref{lem:two-agent-00-limit-loss-separations} gives
$\theta_{00}^{SE}\in SE_{00}\setminus SE_{00}^{Ord}$ and hence
$\theta_{00}^{SE}\in E_{00}^{Ord}\setminus SE_{00}^{Ord}$.
Again these witnesses give the two incomparability results. Lemma~\ref{lem:two-agent-00-limit-loss-separations} gives
$\theta_{00}^{<\omega}\in
S_{00}^{<\omega}\setminus E_{00}^{Ord}$, proving the strictness of
$E_{00}^{Ord}\Rightarrow S_{00}^{<\omega}$.

For the $10$ and $01$ cases,
Lemma~\ref{lem:two-agent-reversal}, together with the
ordinal equivalences above, gives
$E_{ij}\subsetneq E_{ij}^{Ord}$.
Lemma~\ref{lem:two-agent-e11-e01-oscillation} gives
$SE_{01}\subsetneq E_{01}$.

    The single-agent equivalences follow from
    Lemma~\ref{lem:single-agent_collapse}.
\end{proof}
In the multi-agent case with at least two agents, the above theorem shows that strong eventuality implies eventuality in all the four cases, both in the finite and transfinite cases. In the finite case, the converse holds only for the 10 case. In the transfinite cases, the converse holds only in the reversal cases.

Also, finite eventuality implies transfinite eventuality but the converse does not hold in all the four cases. On the other hand, finite strong eventuality does not necessarily imply transfinite strong eventuality: finite strong eventuality implies transfinite strong eventuality only on the reversal side. Furthermore, transfinite strong eventuality does not imply finite strong eventuality in all the four cases.

In short, ``strong'' implies ``weak'' while ``finite'' does not necessarily imply ``transfinite'' and vice versa. Note also that the transfinite fixed-point views are equivalent to the transfinite strong eventual notions.
\begin{remark}[Extension to richer languages]\label{rem:extension-to-richer-languages}
Theorems~\ref{thm:finite-classification}
and~\ref{thm:transfinite-classification} remain unchanged for
$\varphi\in\mathcal{L}_{\mathrm{BPAL}}$ by Lemma~\ref{lem:bpal-reduction-axioms}.

For $\varphi\in\mathcal{L}_C$ or
$\varphi\in\mathcal{L}_{\mathrm{BPALC}}$, all implications and
equivalences in the two theorems whose proofs do not use compactness
remain valid. More precisely, the following modifications are needed.

For Theorem~\ref{thm:finite-classification}, the equivalences between
the eventual conditions and their uniform finite-bound
characterizations are replaced by the corresponding implications
from the uniform finite-bound conditions to the eventual conditions.
Thus, for $(i,j)\in\{(1,1),(0,1),(0,0)\}$,
\[
    \exists N\geq1\;
    \models
    \varphi^i\to
    \bigvee_{n=1}^{N}[\uparrow\varphi]^n\varphi^j
    \quad\Longrightarrow\quad
    E_{ij}(\varphi),
\]
while the converse is not asserted. The remaining limit and
fixed-point characterizations in these three items remain unchanged.

In the $(1,0)$ case, the corresponding part of the theorem becomes
\[
\begin{aligned}
&\exists N\geq1\;\models
    [\uparrow\varphi]^N
    \left(
        \neg\varphi\land
        \bigwedge_{i\in G}\Box_i\bot
    \right)\\
&\qquad\Longrightarrow
    SE_{10}(\varphi)
    \Longleftrightarrow
    \forall M,w\,
    \lim_{n\to\infty}\sigma_n^{M,w}=0\\
&\qquad\Longrightarrow
    E_{10}(\varphi)
    \Longleftrightarrow
    \forall M,w\,
    \liminf_{n\to\infty}\sigma_n^{M,w}=0.
\end{aligned}
\]
The implication from $E_{10}$ to always informativeness when true,
together with the subsequent fixed-point equivalences in
Theorem~\ref{thm:finite-classification}(2), also remains valid.
In particular, neither $E_{10}\Rightarrow SE_{10}$ nor the implication
from $E_{10}$ to the uniform terminal-stage condition above is asserted.

For Theorem~\ref{thm:transfinite-classification}, the only
modification is that
$SE_{10}\Leftrightarrow E_{10}$ is replaced by
$SE_{10}\Rightarrow E_{10}$. All the other displayed implications,
equivalences, strictness results, and incomparability results remain
valid. No strictness claim is made here for the additional implication
$SE_{10}\Rightarrow E_{10}$.

In the single-agent case, the original statements remain unchanged
also for $\mathcal{L}_{\mathrm{BPALC}}$, since on single-agent K45
frames $C_{\{a\}}\psi$ is equivalent to $\Box_a\psi$, and the BPAL
reduction axioms then reduce every such formula to a basic epistemic
formula.
\end{remark}
\section{Conclusion}\label{sec:conclusion}
We introduced eventual and strong eventual notions for success, self-refutation, true lies, and impossible lies, as well as their transfinite versions. The first theorem shows that strong eventual notions are characterized by limit while eventual notions are characterized by limsup/liminf and the uniform bound conditions. In particular, iterated announcements of any eventually self-refuting formula destroy all the agents' beliefs by the uniform bound. Eventual notions further imply fixed-point views of the Moore sentence and the self-fulfilling sentence. The second theorem shows that strong eventuality implies eventuality within both the finite and transfinite settings whereas ``finite'' does not necessarily imply ``transfinite'' and vice versa. The transfinite fixed-point views are equivalent to transfinite strong eventual notions. In the single-agent case, all the notions in each item in the two theorems are equivalent.

Future directions include analyses in multi-agent KD45 and S5 to see which directions fail and which notions become equivalent. Analyses for $\mathcal{L}_C$ and $\mathcal{L}_{BPALC}$ are left, although most of the results already hold (Remark~\ref{rem:extension-to-richer-languages}). More fine-grained interpretations and the roles of transfinite iterated announcements are also left (section~\ref{sec:basic-definitions-properties}). Connections between our results and unknowability/unbelievability~\citep{Yamada2026Unknowability}, and the idea of $\varphi$ being commonly unbelievable, which was briefly mentioned in Section~\ref{sec:eventual-strong-eventual-notions}, are also worth exploring. We could relate our results with the classification results of $\sigma$-validity~\citep{Yamada2026}.

\backmatter
\bmhead{Acknowledgements}
I thank Ryo Kashima and Koki Okura for their comments and feedback during seminars.

The author used GPT-5.6 Sol (Ultra) and GPT-6 Astra for reasoning, literature search, coding, drawing figures, and proofreading. Some of the proofs were written by these models but were thoroughly checked and modified by the author, who takes full responsibility for the final content. 

\section*{Statements and Declarations}
\bmhead{Funding}
This research was supported by the Science Tokyo Support Program for Doctoral Students, funded by the Universities for International Research Excellence.
\bmhead{Competing interests}
The author has no competing interests to declare.
\bmhead{Author contributions}
The author is the sole author of the manuscript.
\bmhead{Data availability}
No datasets were generated or analyzed during the current study.
\bibliography{bibtex_file}

@incollection{Holliday2010-HOLMPI-2,
	author = {Wesley H. Holliday and Thomas F. Icard},
	booktitle = {Advances in Modal Logic 8},
	editor = {Lev Dmitrievich Beklemishev and Valentin Goranko and Valentin Shehtman},
	pages = {178--199},
	publisher = {College Publications},
	title = {Moorean Phenomena in Epistemic Logic},
	year = {2010},
  url = {https://www.aiml.net/volumes/volume8/Holliday-Icard.pdf}
}

@article{AgotnesVanDitmarschWang2018,
  author  = {Thomas {\AA}gotnes and Hans van Ditmarsch and Yanjing Wang},
  title   = {True Lies},
  journal = {Synthese},
  volume  = {195},
  number  = {10},
  pages   = {4581--4615},
  year    = {2018},
  doi     = {10.1007/s11229-017-1423-y},
  url     = {https://doi.org/10.1007/s11229-017-1423-y}
}

@article{Miller2005-MILTUO-11,
	author = {Joseph S. Miller and Lawrence S. Moss},
	doi = {10.1007/s11225-005-3612-9},
	journal = {Studia Logica},
	number = {3},
	pages = {373--407},
	publisher = {Kluwer Academic Publishers},
	title = {The Undecidability of Iterated Modal Relativization},
	volume = {79},
	year = {2005}
}

@misc{Yamada2026,
  author        = {Yamada, Eiji},
  title         = {Classification of {$\sigma$}-Validity in Iterated Announcements},
  year          = {2026},
  eprint        = {2607.04685},
  archivePrefix = {arXiv},
  primaryClass  = {cs.LO},
  doi           = {10.48550/arXiv.2607.04685},
  url           = {https://arxiv.org/abs/2607.04685},
  note          = {Version 3, 24 August 2026}
}

@misc{Yamada2026Unknowability,
  author        = {Yamada, Eiji},
  title         = {The Sources of Unknowability and Self-refutation in Epistemic and Dynamic Epistemic Logic},
  year          = {2026},
  eprint        = {2609.21317},
  archivePrefix = {arXiv},
  primaryClass  = {cs.LO},
  doi           = {10.48550/arXiv.2609.21317},
  url           = {https://arxiv.org/abs/2609.21317},
  note          = {Version 1, 18 September 2026}
}

@incollection{Moore1942,
  author    = {Moore, G. E.},
  title     = {A Reply to My Critics},
  booktitle = {The Philosophy of G. E. Moore},
  editor    = {Schilpp, Paul Arthur},
  series    = {The Library of Living Philosophers},
  volume    = {4},
  pages     = {535--677},
  publisher = {Northwestern University Press},
  address   = {Evanston, IL},
  year      = {1942}
}

@book{Wittgenstein1953,
  author     = {Wittgenstein, Ludwig},
  title      = {Philosophical Investigations},
  translator = {Anscombe, G. E. M.},
  publisher  = {Basil Blackwell},
  address    = {Oxford},
  year       = {1953},
  note       = {See Part II, Section x}
}

@book{Hintikka1962,
  author    = {Hintikka, Jaakko},
  title     = {Knowledge and Belief: An Introduction to the Logic of the Two Notions},
  publisher = {Cornell University Press},
  address   = {Ithaca, NY},
  year      = {1962}
}

@article{plaza2007,
author    = {Jan Plaza},
title     = {Logics of Public Communications},
journal   = {Synthese},
volume    = {158},
pages     = {165--179},
year      = {2007},
doi       = {10.1007/s11229-007-9168-7}
}

@article{GerbrandyGroeneveld1997,
  author  = {Gerbrandy, Jelle and Groeneveld, Willem},
  title   = {Reasoning about Information Change},
  journal = {Journal of Logic, Language and Information},
  volume  = {6},
  number  = {2},
  pages   = {147--169},
  year    = {1997},
  doi     = {10.1023/A:1008222603071},
  url     = {https://doi.org/10.1023/A:1008222603071}
}

@inproceedings{KleinRendsvig2017,
  author    = {Dominik Klein and Rasmus K. Rendsvig},
  title     = {Convergence, Continuity and Recurrence in Dynamic Epistemic Logic},
  booktitle = {Logic, Rationality, and Interaction (LORI 2017)},
  series    = {Lecture Notes in Computer Science},
  volume    = {10455},
  pages     = {71--85},
  year      = {2017},
  publisher = {Springer},
  doi       = {10.1007/978-3-662-55665-8_6}
}

@techreport{Sadzik2006,
  author      = {Tomasz Sadzik},
  title       = {Exploring the Iterated Update Universe},
  institution = {Institute for Logic, Language and Computation (ILLC), University of Amsterdam},
  year        = {2006},
  number      = {PP-2006-26},
  month       = {March},
  type        = {ILLC Prepublication Series},
  url         = {https://illc.uva.nl/Research/Publications/Reports/PP-2006-26.text.pdf}
}

@article{vanBenthem2007,
  author  = {Johan van Benthem},
  title   = {Rational Dynamics and Epistemic Logic in Games},
  journal = {International Game Theory Review},
  volume  = {9},
  number  = {1},
  pages   = {13--45},
  year    = {2007},
  doi     = {10.1142/S0219198907001254}
}

@article{Balbiani2008,
  author  = {Balbiani, Philippe and Baltag, Alexandru and van Ditmarsch, Hans and Herzig, Andreas and Hoshi, Tomohiro and de Lima, Tiago},
  title   = {`{Knowable}' as `known after an announcement'},
  journal = {The Review of Symbolic Logic},
  volume  = {1},
  number  = {3},
  pages   = {305--334},
  year    = {2008},
  doi     = {10.1017/S1755020308080210},
  url     = {https://doi.org/10.1017/S1755020308080210}
}
\end{document}